\documentclass[11pt,notitlepage,onecolumn]{article} 
  \usepackage[left=0.75in,right=0.75in,top=0.75in]{geometry}
  \usepackage{hyperref}
  \usepackage{amsmath}
  \usepackage{amsthm}
  \usepackage{amsfonts}
  \usepackage{amssymb}
  \usepackage{color}
  \usepackage[caption=false]{subfig}

  \usepackage{enumerate}
  \usepackage{mathtools}
  \mathtoolsset{showonlyrefs}
  \usepackage{algorithm,algorithmicx}
  \usepackage[noend]{algpseudocode}

  \newtheorem{thm}{Theorem}[section]
  
  \newtheorem{lem}[thm]{Lemma}
  
  \newtheorem{cor}[thm]{Corollary}

  \theoremstyle{definition}

  \newcommand{\tr}{\operatorname{tr}}

  \newcommand*{\cX}{\mathcal{X}}
  \newcommand*{\cW}{\mathcal{W}}
  
  \newcommand{\cS}{\mathcal{S}}

  \usepackage[pdftex]{graphicx}
  \usepackage{epstopdf}

  \newcommand*{\cH}{\mathcal{H}}
  \newcommand*{\cY}{\mathcal{Y}}

  \newcommand{\ket}[1]{|#1\rangle}
  \newcommand{\bra}[1]{\langle#1|}
  
  \newcommand{\proj}[1]{|#1\rangle\langle#1|}

  \newcommand*{\Vex}{\mathsf{Vex}}
  
  \newcommand*{\cN}{\mathcal{N}}
  \newcommand*{\cE}{\mathcal{E}}
  \newcommand*{\cF}{\mathcal{F}}
  \newcommand*{\cB}{\mathcal{B}}
  \newcommand*{\cD}{\mathcal{D}}
  \newcommand*{\cP}{\mathcal{P}}
  \newcommand*{\cG}{\mathcal{G}}
  
  \DeclareMathOperator{\argmin}{\arg\min}
  \DeclareMathOperator{\argmax}{\arg\max}

  \algrenewcommand\alglinenumber[1]{ }
  \algrenewcommand\algorithmicrequire{\textbf{Precondition:}}
  \algrenewcommand\algorithmicensure{\textbf{Postcondition:}}

\newcommand*{\selfadjointops}{\mathcal{B}_{\mathsf{sa}}}

\begin{document}

  \title{\LARGE Jointly constrained semidefinite bilinear programming \\
  with an application to Dobrushin curves}
  
  \author{Stefan Huber, Robert K\"{o}nig and Marco Tomamichel}
  \maketitle
  
  \begin{abstract}
  We propose a branch-and-bound algorithm for minimizing a 
  bilinear functional of the form 
  \[
  f(X,Y) = \tr((X\otimes Y)Q)+\tr(AX)+\tr(BY) ,
  \]
  of pairs of Hermitian matrices~$(X,Y)$ restricted by joint semidefinite programming constraints. The functional is  parametrized by self-adjoint matrices $Q$, $A$ and $B$.
  This problem generalizes that of a bilinear program, where $X$ and $Y$ belong to  polyhedra. The algorithm converges to a global optimum and yields upper and lower bounds on its value in every step.
  Various problems in quantum information theory can be expressed in this form.
  As an example application, we compute Dobrushin curves of quantum channels, giving  upper bounds on classical coding with energy constraints.
  \end{abstract}

  \section{Introduction}
  The bilinear program of the form
  \begin{align}
  \min_{(x,y)\in X\times Y} x^TQy+v^Tx+w^Ty\label{eq:basicbilinearprogram} ,
  \end{align}
  where $X\subset \mathbb{R}^p$ and $Y\subset\mathbb{R}^q$ are polyhedra and $v \in \mathbb{R}^p$, $w \in \mathbb{R}^q$, $Q \in \mathbb{R}^{p \times q}$, is among the 
  most well-studied optimization problems. 
  One of its first appearances is in
  the formulation of certain two nonzero-sum games studied by 
  Nash~\cite{nash51}.  The optimization problem~\eqref{eq:basicbilinearprogram} has
  various applications in operations research and information theory, including network flow problems, dynamic Markovian assignment problems,
  and dynamic production problems\,---\,see~\cite{konno71} and~\cite{Floudas1995} for a discussion of a number of these. Several natural generalizations of the problem~\eqref{eq:basicbilinearprogram} exist. In particular,  a {\em biconvex problem} is of the form $\min_{(x,y)\in \cS}f(x,y)$ where~$\cS$ and $f$ are biconvex, i.e., $\cS_{x_0}=\{y\in\mathbb{R}^q\ |\ (x_0,y)\in\cS\}$ and $f(x_0,\cdot)$ are convex for every $x_0\in\mathbb{R}^p$, and similarly for $y_0\in\mathbb{R}^q$. 
  We refer to~\cite{Gorski2007} for a review of biconvex problems (see also~\cite{FaizAlKhayyal92}).

  Here we consider a different generalization of~\eqref{eq:basicbilinearprogram} which pertains to problems in quantum information theory. We refer to it as {\it jointly constrained semidefinite bilinear programming}. In this generalization, the vectors $x\in\mathbb{R}^p$ and $y\in\mathbb{R}^q$ are replaced by self-adjoint operators $X\in\selfadjointops(\mathbb{C}^p)$ and $Y\in\selfadjointops(\mathbb{C}^q)$ satisfying certain semidefinite programming (SDP) constraints. The bilinear form of the objective function is retained, leading to 
  \begin{align}
  \min_{(X,Y)\in\cS} \tr((X\otimes Y)Q)+\tr(AX)+\tr(BY)\ .\label{eq:mainoptimizationproblem}
  \end{align}
  The problem~\eqref{eq:mainoptimizationproblem} is thus fully specified by
  a subset $\cS\subset \selfadjointops(\mathbb{C}^p)\times\selfadjointops(\mathbb{C}^q)$ 
  of pairs $(X,Y)$ defined by a family of SDP constraints,
  as well as self-adjoint operators~$Q\in\selfadjointops(\mathbb{C}^p\otimes\mathbb{C}^q)$, $A\in\selfadjointops(\mathbb{C}^p)$
  and $B\in\selfadjointops(\mathbb{C}^q)$.

  The problem~\eqref{eq:mainoptimizationproblem} appears in various forms
  throughout quantum information theory \footnote{This problem is not to be confused with the recently introduced quantum bilinear programs~\cite{Berta_2016}.}.
  For example, in entanglement distribution and quantum communication, one seeks to
  generate entanglement between a reference system~$R$
  and a system~$S$ transmitted through a noisy channel modeled by a completely positive trace-preserving (CPTP) map~$\cN:\cB(\cH_A)\rightarrow\cB(\cH_B)$. 
  A key figure of merit in this context is the entanglement fidelity~\cite{schumacher96}
  \begin{align}
  \max_{(\cE,\cD)}\bra{\Psi}((\cD\circ \cN\circ\cE\otimes\mathsf{id}_R)(\proj{\Psi}))\ket{\Psi}\label{eq:entanglementfidelity}
  \end{align}
  where the optimization is over all encoding CPTP maps $\cE:\cB(\cH_S)\rightarrow\cB(\cH_A)$
  and decoding CPTP maps~$\cD:\cB(\cH_B)\rightarrow\cB(\cH_S)$, and where $\ket{\Psi}\in\cH_S\otimes\cH_R$ is a fixed (maximally entangled) state of the joint system~$SR$. 
  Eq.~\eqref{eq:entanglementfidelity} can be cast in the form~\eqref{eq:mainoptimizationproblem}:
  The set of CPTP  maps $\cE:\cB(\cH_S)\rightarrow\cB(\cH_A)$ can be described by SDP constraints via the Choi-Jamiolkowski isomorphism (and analogously for~$\cD$), and the objective function is bilinear in~$(\cE,\cD)$.

  Another context in which the problem~\eqref{eq:mainoptimizationproblem} appears naturally is the setting of Bell inequalities and quantum games. Consider for example the bipartite case, where a state $\ket{\Psi}\in\cH_A\otimes\cH_B$ is given. A Bell inequality can be expressed as a lower bound on 
  the expectation value $\bra{\Psi}B\ket{\Psi}$ of a Bell operator
  $B=B(\{E_j\}_{j},\{F_k\}_{k})\in\selfadjointops(\cH_A\otimes\cH_B)$ which depends on
  observables
  $\{E_j\}_{j}$ on $\cH_A$ (corresponding to $A$'s measurement settings)
  and observables $\{F_k\}_{k}$ on $\cH_B$ (corresponding to $B$'s measurement settings). 
   The Bell operator usually depends bilinearly on $(\{E_j\}_{j},\{F_k\}_{k})$: for example, the Bell-CHSH operator involves two measurement settings each and is given by $B(E_0,E_1,F_0,F_1)~=~E_0\otimes (F_0+F_1)+E_1\otimes (F_0-F_1)$. 
  Since an observable $A$ is a self-adjoint operator satisfying $-I\leq A\leq I$, the problem of finding the optimal value 
  \begin{align}
  \max_{\{E_j\}_{j},\{F_k\}_{k}} \bra{\Psi}B(\{E_j\}_{j},\{F_k\}_{k})\ket{\Psi}\ \label{eq:bellinequalityoptimization}
  \end{align}
  optimized over all observables can  directly be cast in the form~\eqref{eq:mainoptimizationproblem}. 

  We note that the form of~\eqref{eq:mainoptimizationproblem} is 
  somewhat more general than what is required in most applications to quantum information theory such as problems~\eqref{eq:entanglementfidelity} and~\eqref{eq:bellinequalityoptimization}. Indeed, in the latter two problems,
  there are no joint constraints (i.e., the set $\cS=\cS_1\times\cS_2$ is a product of two sets, each of which is defined by SDP constraints), and the objective function has no linear terms. In Section~\ref{sec:dobrushin}, we discuss a problem from quantum information theory whose reformulation in terms of~\eqref{eq:mainoptimizationproblem} involves linear terms. 
  
  A useful alternative but equivalent form of the optimization problem is given by
 \begin{align}
   \min_{(X,Y)\in\cS} \tr(X \mathcal{E}(Y)) + \tr (A X ) + \tr( B Y) \label{eq:bilinearprogram2}
 \end{align} 
 where $\mathcal{E}$ is a Hermiticity-preserving operation. The one-to-one correspondence of $\mathcal{E}$ and $Q$ is a consequence of the Choi-Jamiolkowski isomorphism. 
 
  Finally we note that joint constraints allow to consider quadratic optimization problems of the form
  \begin{align}
  \min_{X\in\cS} \tr((X\otimes X)Q)+\tr(AX)\ ,\label{eq:quadraticprogramdef}
  \end{align}
  where $\cS$ is a set defined by SDP constraints, simply by imposing that $X=Y$.
  A basic example is a situation  where  one is interested in the maximal (or minimal) uncertainty
   when measuring a state~$\ket{\Psi}$, using an observable~$X\in\cS$ belonging to a set~$\cS$ specified by SDP constraints. If uncertainty is quantified by the variance~$\mathsf{Var}_\Psi(X):=\bra{\Psi}X^2\ket{\Psi}-\bra{\Psi}X\ket{\Psi}^2$, then the problem~$\min_{X\in\cS}\mathsf{Var}_{\Psi}(X)$ can be recast in the form~\eqref{eq:quadraticprogramdef}.

  \subsection*{The seesaw algorithm in quantum information}
  Given the ubiquity of optimization problems of the form~\eqref{eq:mainoptimizationproblem} in quantum information theory, it is natural to seek algorithms computing its value as well as optimal solutions~$(X,Y)$. A widely used
  and often successful heuristic is referred to as the seesaw or mountain climbing algorithm. It is based on the observation that for every $X_0$  (associated with a feasible point $(X_0,Y_0)\in\cS$), the function 
  $f_{X_0}(Y)=\tr((X_0\otimes Y)Q)+\tr(AX_0)+\tr(BY)$ is
  linear up to the additive constant~$\tr(AX_0)$. Furthermore, the set $\cS_{X_0}=\{Y\in\cB(\mathbb{C}^q)\ |\ (X_0,Y)\in\cS\}$
  can be described by SDP constraints (indeed, we can augment those specifying $\cS$ by  the constraint $X=X_0)$.  Thus,
  $Y_0:=\argmin_{Y\in \cS_{X_0}}f_{X_0}(Y)$ can be found by solving an SDP. The role of $X$ and $Y$ is then interchanged: in a next step, $X_1:=\argmin_{Y\in \cS_{Y_0}}f_{Y_0}(X)$ 
  is computed (where $\cS_{Y_0}$ and $f_{Y_0}$ are defined analogously). Iterating this produces a sequence of pairs $(X_j,Y_j)$. It can be shown that in a finite number of iterations, this sequence converges to a Kuhn-Tucker point~$(\bar{X},\bar{Y})$  of the objective function 
  \begin{align}
  f(X,Y)=\tr((X\otimes Y)Q)+\tr(VX)+\tr(WY)\ ,
  \end{align}
  see~\cite[Prop.~2.3]{Konno1976} for an analysis of the analogous algorithm for bilinear programs, and~\cite[Theorem~5]{Nahapetyan2009}
   as well as~\cite{horst2000introduction} for conditions guaranteeing that this point is a local optimimum. Thus, while 
   it is not generally the case that $(\bar{X},\bar{Y})$ is a local (let alone global) optimum, this algorithm may\,---\,for a suitable choice of initial points $(X_0,Y_0)$ indeed result in a global solution. It should be emphasized, however, that even in that  case, global optimality needs to be established by other means.
   
   Despite these limitations, the seesaw algorithm has been quite popular and has been successfully applied in quantum information theory. Its use in the context of Bell inequalities was
  first discussed in~\cite[Section 5.1]{wernerwolf01}.
  In the context of error correction, the maximization of fidelity optimized over encoder and decoder has been investigated numerically using the seesaw algorithm, see~\cite{reimpellwerner05} and~\cite{KosutLidar2009}.
  More recently, a variant of the seesaw algorithm (involving a trilinear function) was used in~\cite{VertesiBrunner} to optimize  the value of a Bell inequality over PPT-states, yielding a counterexample to Peres' conjecture~\cite{peres96} that non-distillable states are local.

   To date, the seesaw algorithm appears to be the only procedure for optimization problems of the form~\eqref{eq:mainoptimizationproblem} which has been used in the context of quantum information. This is in sharp contrast to 
  the bilinear program~\eqref{eq:basicbilinearprogram}, for which a variety of algorithms have been proposed. This includes   cutting plane algorithms~\cite{Konno1976,vaishshettycuttingplane,Sherali1980}, branch-and-bound algorithms~\cite{Falk1973,kahfalk83}, extreme point ranking procedures~\cite{cabotfrancis} and methods based on polyhedral annexation~\cite{vaishshettypolytope} (see~\cite{Floudas1995} for a review). 

   \subsection*{Our contribution}
  Our main contribution is a branch-and-bound algorithm for the jointly constrained semidefinite bilinear program~\eqref{eq:mainoptimizationproblem}. It is a generalization of the branch-and-bound algorithm  by Al-Khayyal and Falk~\cite{kahfalk83} which we review in Section~\ref{sec:jointlyconstrainedbiconvex}. Roughly, our algorithm proceeds by iteratively solving semidefinite programs providing upper and lower bounds on
  the value of~\eqref{eq:mainoptimizationproblem}. 
  Following standard arguments (see e.g.,~\cite{kahfalk83}), it can be shown to produce a sequence of feasible points $(X_i,Y_i)\in\cS$ such that $f(X_i,Y_i)$  converges  to the global optimum~\eqref{eq:mainoptimizationproblem}. More importantly, it provides\,---\,at each stage~$i$\,---\,a
   bound on the deviation of $f(X_i,Y_i)$ from the optimum~\eqref{eq:mainoptimizationproblem}.

  To illustrate the practical use of our algorithm, we  apply it to a problem in quantum information theory: we  compute so-called Dobrushin curves for quantum channels. These give upper bounds on optimal codes for classical information in a scenario where the noise acts repeatedly.

  \subsubsection*{Outline of the paper}
  In Section~\ref{sec:bandb}, we briefly review branch-and-bound algorithms and discuss the algorithm by Al-Khayyal and Falk~\cite{kahfalk83} for solving jointly constrained biconvex programs. In Section~\ref{sec:mainalgorithm}, we give our algorithm for jointly constrained semidefinite bilinear programs. Finally, in Section~\ref{sec:dobrushin}, we discuss the application to Dobrushin curves.

  \section{Branch-and-bound algorithms\label{sec:bandb}}
  In this section, we  review the branch-and-bound algorithm of Al-Khayyal and Falk~\cite{kahfalk83} to solve  jointly constrained biconvex programs. We introduce
  the jointly constrained biconvex programming problem, and then give a  description of the algorithm of~\cite{kahfalk83}.
   
  \subsection{Jointly constrained biconvex programming\label{sec:jointlyconstrainedbiconvex}}
  
  To define jointly constrained biconvex programs, let $\cS\subset \mathbb{R}^n\times\mathbb{R}^n$ be a non-empty, closed and convex set. For later convenience, also let  $\cD~=~\Omega(\ell,L,m,M)~\subset~\mathbb{R}^n \times \mathbb{R}^n$ be the (product of) hyperrectangle(s) defined in terms of the vectors~$\ell,L,m,M \in \mathbb{R}^n$ as 
  \begin{align}\Omega(\ell,L,m,M) = \{(x,y) \in \mathbb{R}^n \times \mathbb{R}^n \big| \  \ell_i \leq x_i \leq L_i, m_i \leq y_i \leq M_i \textrm{ for all } i = 1, \dots, n\}\ .\label{eq:hyperrectangle} 
  \end{align}
  Furthermore, let $f,g: \cD\rightarrow\mathbb{R}$ be
  such that their restrictions to $\cS\cap \cD$ are convex. The  jointly constrained biconvex program is  the problem
  \begin{align}
      \min_{(x,y)\in\cS \cap \cD }& \;F(x,y)\qquad\textrm{where }\qquad  F(x,y):= f(x) + x^T y + g(y)\ .
  \label{eq:biconvexprogram}
  \end{align}
  The set $\cS$  permits to include joint constraints on the vectors $x$ and $y$. We note that although the restrictions $F(\cdot,y)$ and $F(x,\cdot)$ are convex for each $(x,y)\in\cS\cap \cD$\,---\,a property referred to as biconvexity\,---\,the problem Eq.~\eqref{eq:biconvexprogram} itself is non-convex.

Eq.~\eqref{eq:biconvexprogram} is a generalization of the  bilinear program \eqref{eq:basicbilinearprogram} discussed in the introduction. 
  Indeed,  Eq.~\eqref{eq:basicbilinearprogram} can be transformed into a problem of the form \eqref{eq:biconvexprogram}
  by replacing $x^T (Q y)$ by $x^T z$ and adding the linear constraint $z = Qy$ to the defining constraints of $\cS$.

  \subsubsection{Obtaining lower bounds on the biconvex program\label{sec:lowerboundsbiconvex}}
  Being non-convex, Eq.~\eqref{eq:biconvexprogram} cannot directly be addressed  with convex solvers.   However, one can construct a convex problem whose solution gives a lower bound on the value of Eq.~\eqref{eq:biconvexprogram}. This relies on the  concept of the convex envelope $\Vex_{\cD}F:\cD\rightarrow\mathbb{R}$ of a function $F:\cD\rightarrow \mathbb{R}$, where $\cD\subset \mathbb{R}^n \times \mathbb{R}^n$. It is defined as the pointwise supremum of all convex functions
  underestimating $F$ over $\cD$, i.e., 
  \begin{align}
    (\Vex_{\cD}F)(x,y):=\sup_{
  \substack{G: \cD\rightarrow \mathbb{R}\textrm{ convex}\\
G(v,w)\leq F(v,w)\textrm{ for all }(v,w)\in \cD}}
G(x,y)\qquad\textrm{ for all }(x,y)\in \cD\  ,
  \end{align}
  see~\cite{falk69} for more details.

  To compute the
  convex envelope $\Vex_{\cS\cap \cD}F$ of
  the objective function in Eq.~\eqref{eq:biconvexprogram}, 
  where $\cD=\Omega$ is a hyperrectangle,   one uses the fact that  the convex envelope
  of the function $(x,y)\mapsto x^Ty$ over a hyperrectangle~$\Omega$ (cf.~\eqref{eq:hyperrectangle}) is (see~\cite[Corollary to Theorems 2 and 3]{kahfalk83})
    \begin{align}
    \Vex_{\Omega} (x^T y) &= \sum_{i=1}^n \Vex_{\Omega_i} (x_iy_i) \qquad&\textrm{ for all } (x,y) \in \Omega\ ,
   \end{align}
    where 
\begin{align}
  \Omega_i := \{ (x, y) \in \mathbb{R} \times \mathbb{R}: \ell_i \leq x \leq L_i, m_i \leq y \leq M_i\}=[\ell_i,L_i]\times [m_i\times M_i]\subset \mathbb{R}\times \mathbb{R}
\end{align}
is the projection of the hyperrectangle $\Omega$ onto the $i$-th pair of coordinates for $i = 1, \dots, n$. The convex envelope of the function $(x,y)\mapsto xy$ 
over~$\Omega_i\subset\mathbb{R}\times\mathbb{R}$ is (see \cite[Theorem 2]{kahfalk83})
  \begin{align}
    \Vex_{\Omega_i}(x y) = \max\{m_i x + \ell_i y - \ell_i m_i, M_i x + L_i y_i - L_i M_i \}\ .\label{eq:convexenvelopeinterval}
  \end{align}
  Hence the convex envelope of $F$ in Eq.~\eqref{eq:biconvexprogram} over a hyperrectangle~$\Omega$ has a simple expression, i.e.\
  \begin{align}
  (\Vex_{\Omega}F)(x,y)=f(x) + \sum_{i=1}^n \Vex_{\Omega_i} (x_iy_i) + g(y) \label{eq:convexenvgeneral} \ .
  \end{align}
  In addition to being a convex underestimator for~$F$, the function $\Vex_{\Omega}F$ has the important property that it agrees with $F$ on the boundary  
\begin{equation}
  \partial\Omega = \Omega \setminus  \{ (x,y) \in \mathbb{R}^n \times \mathbb{R}^n\ \big|\ \ell_i < x_i < L_i, m_i < y_i < M_i \textrm{ for all } i = 1, \dots, n\}\ 
\end{equation}
of $\Omega$. Indeed, this follows from the analogous property 
$\Vex_\Omega (x^Ty)=x^Ty$ for all $(x,y)\in \partial\Omega$ for the function $(x,y)\mapsto x^Ty$, see~\cite[Theorem 3]{kahfalk83}. 

Note that, given~\eqref{eq:convexenvgeneral}, the problem
$\min_{(x,y)\in \cS\cap \Omega}(\Vex_\Omega F)(x,y)$ 
can be treated using a convex solver, giving a value $\underline{\alpha}$.
 In particular, since $\Vex_\Omega F$ underestimates~$F$ over $\Omega$ (and hence over $\cS\cap \Omega$), the value~$\underline{\alpha}(\Omega)$ provides a (global) lower bound on the problem~\eqref{eq:biconvexprogram}. Trivially, any point $(x^*,y^*)\in\cS\cap \Omega$ (including one that achieves the minimum 
 of $\Vex_\Omega F$)  also provides an upper bound~$\overline{\alpha}(\Omega)=F(x^*,y^*)$ on~\eqref{eq:biconvexprogram}. Thus, we obtain
 \begin{align}
 \underline{\alpha}(\Omega)\leq \min_{(x,y)\in\cS\cap \Omega}F(x,y)\leq \overline{\alpha}(\Omega)\ .
 \end{align}

 We note that the reasoning here applies to any hyperrectangle~$\Omega$. 
 
  \subsubsection{The branch-and-bound algorithm for the biconvex program}\label{sec:biconvex}
  
  We can now sketch the algorithm of~\cite{kahfalk83} which solves problem \eqref{eq:biconvexprogram}. The algorithm takes as input the functions $f$ and $g$, the set $\cS \subset \mathbb{R}^n \times \mathbb{R}^n$ 
   and the vectors  $\ell,L,m,M \in \mathbb{R}^n$ specifying the hyperrectangle $\Omega~=~\Omega(\ell,L,m,M)$. In addition, it accepts a  desired precision $\epsilon > 0$. The algorithm returns a value $\overline{\alpha}$ and a  point $(x^*,y^*) \in \cS \cap \Omega$ such that
  \begin{equation}
    \overline{\alpha} = F(x^*,y^*)  \leq \left(\min_{(x,y) \in \cS \cap \Omega} F(x,y)\right) + \epsilon\ .
    \label{eq:outputalg}
  \end{equation}

   Given the technique for finding lower bounds on~$\min_{(x,y)\in\cS\cap \cD}F(x,y)$ discussed in Section~\ref{sec:lowerboundsbiconvex}, the main idea underlying the algorithm is to apply this strategy to increasingly smaller hyperrectangle~$\Omega$ (which together form a partition of~$D$). The respective upper and lower bounds for each hyperrectangle give (global) upper and lower bounds on the biconvex problem~\eqref{eq:biconvexprogram}.

   More precisely, the algorithm keeps track of a finite list $\cP$ of 
   hyperrectangles which form a partition of~$\cD$. In addition,
   for each $\Omega\in\cP$, the values $(\underline{\alpha}(\Omega),\overline{\alpha}(\Omega))$ are computed and kept track of such that
   \begin{align}
   \underline{\alpha}(\Omega)\leq \min_{(x,y)\in\cS\cap \Omega}F(x,y)\leq \overline{\alpha}(\Omega)\ .
   \end{align}
   Finally, $z(\Omega)=(x,y)\in\cS\cap \Omega$ will be an element such that $F(z)=\overline{\alpha}(\Omega)$. 
   
   As a consequence, the quantities
  \begin{align}
    \underline{\alpha}(\cP)=\min_{\Omega\in\cP}\underline{\alpha}(\Omega)\qquad \textrm{and} \qquad
  \overline{\alpha}(\cP)=\min_{\Omega\in\cP}\overline{\alpha}(\Omega)\ 
 \end{align}
 constitute global upper and lower bounds on the problem Eq.~\eqref{eq:biconvexprogram},  that is,
 \begin{align}
  \underline{\alpha}(\cP)\leq \min_{(x,y)\in\cS\cap \Omega}F(x,y)\leq \overline{\alpha}(\cP)\ .
 \end{align}
 As soon as $\overline{\alpha}(\cP)-
 \underline{\alpha}(\cP)\leq \epsilon$, the algorithm returns $\overline{\alpha}(\cP)$ and $(x^*,y^*)=z(\Omega')$, where $\Omega'\in\cP$ is such that $F(z(\Omega'))=\overline{\alpha}(\cP)$.

  Recall that $\cD$ itself is a hyperrectangle by definition of the problem. Consequently, we begin with the (trivial) partition $\cP=\{\cD\}$. Since an algorithm for computing bounds $(\underline{\alpha}(\Omega)$, $\overline{\alpha}(\Omega))$ and points $z(\Omega)\in\cS\cap \Omega$ for any hyperrectangle~$\Omega$ is already constructed, it remains to specify how $\cP$ is successively refined. 
  
  Assume that the algorithm has not returned a solution yet, i.e., that 
    \begin{align}
    \overline{\alpha}(\cP)-\underline{\alpha}(\cP)\geq \epsilon\ .\label{eq:terminationconditionnot}
    \end{align} 
 The idea here is to try to improve the worst lower bound. That is,
 pick a hyperrectangle~$\Omega\in\cP$ such that $\underline{\alpha}(\cP)=\underline{\alpha}(\Omega)$. 
 \begin{figure}\centering
  \includegraphics{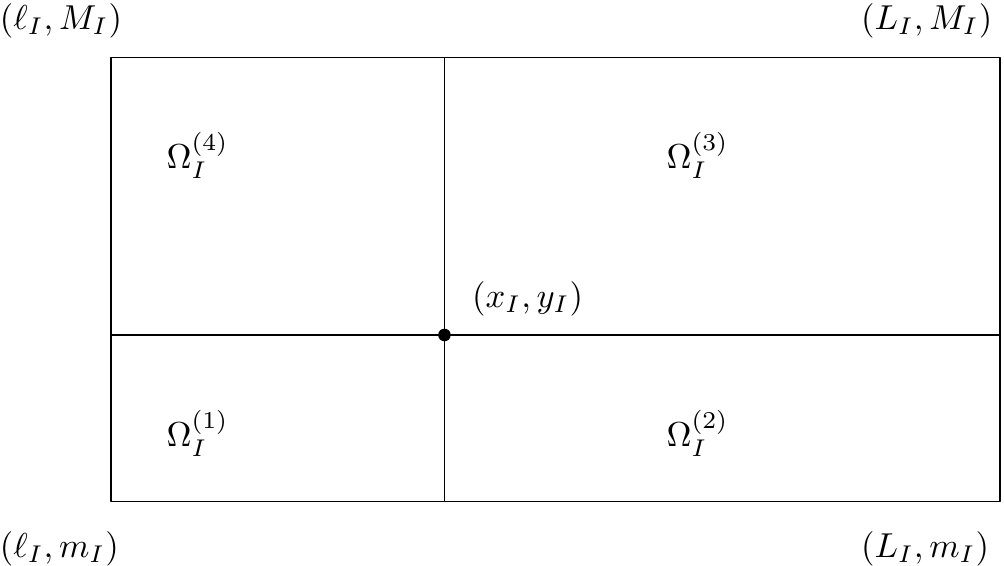}
  \caption{The creation of hyperrectangles $\Omega^{(1)}, \Omega^{(2)}, \Omega^{(3)}, \Omega^{(4)}$ from $\Omega$. This shows the projection onto the coordinates $(x_I,y_I)$ of these hyperrectangles.\label{fig:hyperrectanglesplit}}
  \end{figure}
 We then subdivide~$\Omega$ in $4$~new hyperrectangles $\Omega^{(1)},\ldots,\Omega^{(4)}$. To do so, observe that Eq.~\eqref{eq:terminationconditionnot} implies that $\overline{\alpha}(\Omega)-\underline{\alpha}(\Omega)\geq \epsilon$. 
 Hence, by definition of $\overline{\alpha}(\Omega)$ and $\underline{\alpha}(\Omega)$, there must exist at least one $i \in \{ 1, \dots, n\}$ such that $\Vex_{\Omega}(x_iy_i)~<~x_i y_i$. We pick the index $I$ which leads to the largest difference between the two sides of this inequality and split up the rectangle $\Omega$ into four subrectangles, arriving at the new hyperrectangles $\{ \Omega^{(j)}\}_{j = 1}^4$. For each $j\in \{1,\ldots,4\}$, the hyperrectangle $\Omega^{(j)}$ is defined by its projections
 \begin{align}
 \Omega_{i}^{(j)} =
 \begin{cases}
 \Omega_i\qquad &\textrm{ if }i\in \{1,\ldots,n\}\backslash \{I\}\\
 A^{(j)}\qquad &\textrm{ for }i=I\ .
 \end{cases}
 \end{align}
 onto pairs of coordinates. 
 Here $\{A^{(j)}\}_{j=1}^4$ is a certain partition of $\Omega_I\subset\mathbb{R}\times\mathbb{R}$ into four subrectangles, as shown in Fig.~\ref{fig:hyperrectanglesplit}.
 The latter is defined by the pair of $I$-th coordinates  $(x_I, y_I)$
 of the point $z(\Omega)=(x,y)$, as shown in Fig.~\ref{fig:hyperrectanglesplit}.  Hence we have constructed a partition $\{ \Omega^{(j)} \}_{j=1}^4$ of $\Omega$ into smaller hyperrectangles. 
 
 These steps are iterated until Eq.~\eqref{eq:terminationconditionnot} is no longer satisfied.   This procedure can be shown to converge to a globally optimal value of the problem~\eqref{eq:biconvexprogram}, as done in~\cite{kahfalk83}.

\section{Jointly constrained semidefinite bilinear programming\label{sec:mainalgorithm}}
  Suppose self-adjoint operators $Q \in \selfadjointops(\mathbb{C}^p \otimes \mathbb{C}^q), A \in \selfadjointops(\mathbb{C}^q), B \in \selfadjointops(\mathbb{C}^p)$ are given.  Here $\selfadjointops(\cH)$ denotes the real vector space of self-adjoint operators
  on a Hilbert space~$\cH$ with  respect to the
  Hilbert-Schmidt inner product $\langle A,B\rangle=\tr(A B)$.
    Define a function $F:\selfadjointops(\mathbb{C}^p) \times \selfadjointops(\mathbb{C}^q)\rightarrow\mathbb{R}$ by
  \begin{align}
  F(X,Y):=\tr\left((X\otimes Y)Q\right)+\tr(A X)+\tr(B Y) \ .\label{eq:sdpproblemfunction}
  \end{align}
    We consider the problem
  \begin{align}
    \inf_{(X,Y)\in \cS} F(X,Y)\ ,
  \label{eq:sdpproblem}
  \end{align}
where $\cS \subset \selfadjointops(\mathbb{C}^p) \times \selfadjointops(\mathbb{C}^q)$ is defined by a family of semidefinite constraints, which may involve both $X$ and $Y$ (in particular,~$\cS$ is convex). We note that the function~$F$ is again biconvex but not convex. We refer to Eq.~\eqref{eq:sdpproblem} as a jointly constrained semidefinite bilinear program.

A first step to construct an algorithm for~\eqref{eq:sdpproblem} is to rephrase it in a form similar to~\eqref{eq:biconvexprogram}. To do so,
let $\{\eta_j\}_{j=1}^{p^2} \subset\selfadjointops(\mathbb{C}^p)$ and $\{\xi_k\}_{k=1}^{q^2} \subset \selfadjointops(\mathbb{C}^q)$ be orthonormal operator bases of the real vector spaces~$\selfadjointops(\mathbb{C}^p)$ and~$\selfadjointops(\mathbb{C}^q)$,  respectively.
It will be convenient to express
operators in terms of coefficients
in bases that are rotated with respect to $\{\eta_j\}_{j=1}^{p^2}$
and $\{\xi_k\}_{k=1}^{q^2}$, with a rotation depending on the objective function.
Consider the $p^2\times q^2$-matrix $U_{j,k}=\tr(Q(\eta_j\otimes\xi_k))$ and let \begin{align}
    U=S \Delta T\qquad\textrm{ where } \Delta \in \mathbb{R}^{p^2 \times q^2}\
    \label{eq:svd}
  \end{align}
    be its singular value decomposition, i.e., $S\in\mathbb{R}^{p^2\times p^2}$ and
     $T\in\mathbb{R}^{q^2\times q^2}$
   are orthogonal, and $\Delta$ has
   the singular values~$\{\sigma_j\}_{j=1}^K$ on the diagonal (here $K\leq \min\{p^2,q^2\}$).
    Define the map
    \begin{align}
      \Gamma:\selfadjointops(\mathbb{C}^p) \times \selfadjointops(\mathbb{C}^q) &\rightarrow \mathbb{R}^{p^2}\times\mathbb{R}^{q^2}\\
  (X,Y) &\mapsto \Gamma(X,Y)=(x(X),y(Y))\ ,\\
    \end{align}
where
  \begin{align}
  x_j(X)=\sum_{k=1}^{p^2}S_{k,j}\tr(X\eta_k)\qquad&\textrm{ and }\qquad
    y_k(Y)=\sum_{\ell=1}^{q^2}T_{k,\ell}\tr(Y\xi_\ell)\ .
  \end{align}
    Let $(a,b)=(x(A),y(B))\in\mathbb{R}^{p^2}\times\mathbb{R}^{q^2}$. Define the function $f:\mathbb{R}^{p^2}\times \mathbb{R}^{q^2}\rightarrow\mathbb{R}$ by
    \begin{align}
    f(x,y):=\sum_{j=1}^K\sigma_j x_j y_j + \sum_{j=1}^{p^2} a_j x_j+\sum_{k=1}^{q^2} b_ky_k\ .\label{eq:fvectorformexplicit}
    \end{align}
    
    Using this construction we can now reduce the matrix problem to an equivalent vector problem:
\begin{lem}
    $\Gamma$ is one-to-one and $F(X,Y)=f(x(X),y(Y))$ for all $(X,Y)\in\cS$. In particular,
  \begin{align}
    \inf_{(X,Y) \in S} F(X,Y) = \inf_{(x,y) \in \Gamma(\cS)}f(x,y)\ .
  \end{align}
  \label{lem:vectorproblem}
\end{lem}
\begin{proof}
Observe that for $(X,Y)$, we have
  \begin{align}
 F(X,Y)=\sum_{j=1}^{p^2} \sum_{k=1}^{q^2} \hat{x}_j\hat{y}_k U_{j,k}+\sum_{j=1}^{p^2} \hat{x}_j \hat{a}_j+\sum_{k=1}^{q^2} \hat{y}_k \hat{b}_k\ ,
  \end{align}
  where $\hat{x}_j=\tr(X\eta_j)$, $\hat{a}_j=\tr(A\eta_j)$ for $j=1,\ldots,p^2$, and
  similarly $\hat{y}_k=\tr(Y\xi_k)$, $\hat{b}_k=\tr(B\xi_k)$ for $k=1,\ldots,q^2$.
  Using the variable substitutions
  \begin{align}
  \begin{matrix}
  x&=&S^T \hat{x}\\
  y&=&T \hat{y}
  \end{matrix}\qquad\textrm{and }\qquad
  \begin{matrix}
    a&=&S^T \hat{a}\\
    b&=&T \hat{b}\ ,
    \end{matrix}
  \end{align}the claim follows.
\end{proof}
Given Lemma~\ref{lem:vectorproblem}, our algorithm proceeds by
first finding a hyperrectangle~$\cD\subset\mathbb{R}^{p^2}\times\mathbb{R}^{q^2}$ that contains the set $\Gamma(\cS)$ (see Section~\ref{sec:boundinghyperrectangle}). We then argue that
 lower bounds on the objective function restricted to hyperrectangles can be computed by solving SDPs (see Section~\ref{sec:sdplowerbound}). A branch-and-bound procedure for the problem~\eqref{eq:sdpproblem} follows.

 For later reference, we give pseudocode of two routines \textsc{ComputeVectorRep} and \textsc{ComputeOperator}, which compute the functions~$\Gamma$ respectively~$\Gamma^{-1}$ appearing in Lemma~\ref{lem:vectorproblem}, see Fig.~\ref{alg:computevectorrep} in Appendix~\ref{app:pseudo}.

\subsection{Finding a bounding hyperrectangle\label{sec:boundinghyperrectangle}}
For $\ell,L\in\mathbb{R}^{p^2}$ and $m,M\in\mathbb{R}^{q^2}$, define the hyperrectangle
\begin{align}
\Omega(\ell,L,m,M) = \{(x,y) \in \mathbb{R}^{p^2} \times \mathbb{R}^{q^2} \big| \  &\ell_j \leq x_j \leq L_j \textrm{ for all } j=1,\dots,p^2 \\
& \textrm{ and } m_k \leq y_k \leq M_k \textrm{ for all } k = 1, \dots,q^2\}\ .
\end{align}

We show the following:
\begin{lem}\label{lem:boundinghyperrectangle}
We can efficiently find  $\ell,L\in\mathbb{R}^{p^2}$ and $m,M\in\mathbb{R}^{q^2}$
such that $\Omega(\ell,L,m,M)$ has minimal volume among all
hyperrectangles $\Omega$ containing the set $\Gamma(\cS)$, where $\Gamma$ is defined as in Lemma~\ref{lem:vectorproblem}. More precisely, we can find such vectors by solving $2(p^2+q^2)$ SDPs in~$(X,Y)\in\cS$.
\end{lem}
\begin{proof}
Clearly, we need to compute
\begin{align}
\begin{matrix}
\ell_j^*&=&\inf_{(X,Y) \in \cS} x_j(X)\;\;&\textrm{ and }\;\;&
  L_j^*&=&\sup_{(X,Y) \in \cS} x_j(Y)\;\;&\textrm{ for }\;\; &j=1,\ldots, p^2\\
  & & & & &\textrm{ as well as }\\
  m_k^*&=&\inf_{(X,Y) \in \cS} y_k(X)\;\;&\textrm{ and }\;\;&
  M_k^*&=&\sup_{(X,Y) \in \cS} y_k(Y)\;\;&\textrm{ for }\;\; &k=1,\ldots, q^2\ .
  \end{matrix}
  \end{align}
  Here we write $\Gamma(X,Y)=(x(X),y(Y))$ as in Lemma~\ref{lem:vectorproblem}. It is easy to see that each of these optimization problems is  an SDP. For example, for each $j\in \{1,\ldots,p^2\}$, we have
  \begin{align}
  \ell_j^*&=\inf_{(X,Y)\in\cS} \sum_{k=1}^{p^2} S_{k,j}\tr(X \eta_k)\
  \end{align}
  and similar reasoning applies to the values $L_j^*,m_k^*$ and $M_k^*$.
\end{proof}

Pseudocode for the associated procedure is given in Fig.~\ref{alg:boundingrec} in Appendix~\ref{app:pseudo}.

\subsection{Obtaining lower bounds on the semidefinite bilinear program\label{sec:sdplowerbound}}
As in Section~\ref{sec:lowerboundsbiconvex}, we
next discuss how to find lower and upper bounds $\underline{\alpha}(\Omega)$, $\overline{\alpha}(\Omega)$ on the objective function~$F(X,Y)$ restricted to the preimage $\Gamma^{-1}(\Omega)$ of a hyperrectangle~$\Omega\subset\mathbb{R}^{p^2}\times\mathbb{R}^{q^2}$. That is, in terms of the
function $f: \mathbb{R}^{p^2} \times \mathbb{R}^{q^2} \rightarrow \mathbb{R}$ defined in Lemma~\ref{lem:vectorproblem}, these values  satisfy
\begin{align}
\underline{\alpha}(\Omega)&\leq
\inf_{(x,y)\in \Gamma(\cS)\cap \Omega} f(x,y)\leq \overline{\alpha}(\Omega)\ .\label{eq:twoprovelowerboundalpha}
\end{align}
For the lower bound, recalling the definition of the convex envelope introduced in Section~\ref{sec:biconvex}, it suffices to compute \begin{align}
\underline{\alpha}(\Omega)&=\inf_{(x,y)\in \Gamma(\cS)\cap \Omega} (\Vex_{\Omega}f)(x,y)\ .\label{eq:problemtooptimize2}
\end{align}
On the other hand, any element $(x^*,y^*)\in \Gamma(\cS)\cap \Omega$ provides an upper bound
$\overline{\alpha}(\Omega)=f(x^*,y^*)$.

To compute Eq.~\eqref{eq:problemtooptimize2}, we proceed in two steps. First, we give an explicit expression for $\Vex_{\Omega}f$.
\begin{lem}\label{lem:explicitconvexenve}
  Let $\Omega = \Omega(\ell,L,m,M)\subset \mathbb{R}^{p^2} \times \mathbb{R}^{q^2}$ be a hyperrectangle and $f: \mathbb{R}^{p^2} \times \mathbb{R}^{q^2} \rightarrow \mathbb{R}$ as in Lemma~\ref{lem:vectorproblem}. Then the convex envelope of $f$ over $\Omega$ is given by
\begin{align}
  (\Vex_{\Omega}f)(x,y) = \sum_{j=1}^K \max\{h_j^0(x_j,y_j), h_j^1(x_j,y_j)\} + \sum_{j=1}^{p^2} a_j x_j + \sum_{k=1}^{q^2} b_ky_k\ ,
\end{align}
where
\begin{align}
\begin{matrix}
h_j^0(x_j,y_j)& =& \sigma_j \left(m_j x_j + \ell_j y_j - \ell_j m_j \right)\qquad&\textrm{ and }\\
h_j^1(x_j,y_j) &= &\sigma_j \left(M_j x_j + L_j y_j - L_j M_j \right)\ .
\end{matrix}\label{eq:hfunctions}
\end{align}
for $j=1,\ldots, K$.
\end{lem}

\begin{proof}
By definition of $f$ (see Lemma~\ref{lem:vectorproblem}) and
calculations analogous to those discussed in Section~\ref{sec:lowerboundsbiconvex}, we have
\begin{align}
  (\Vex_{\Omega}f)(x,y)&=\sum_{j=1}^K  \Vex_{\Omega_j} \left(\sigma_jx_jy_j\right) + \sum_{j=1}^{p^2} a_jx_j+\sum_{k=1}^{q^2} b_ky_k\\
  &=\sum_{j=1}^K
  \max\{\sigma_j\cdot(m_jx_j+\ell_jy_j-\ell_jm_j),\sigma_j\cdot (M_jx_j+L_jy_j-L_jM_j)\} + \sum_{j=1}^{p^2} a_j x_j + \sum_{k=1}^{q^2} b_k y_k\\
  &= \sum_{j=1}^K \max\{h_j^0(x_j,y_j), h_j^1(x_j,y_j)\} + \sum_{j=1}^{p^2} a_j x_j + \sum_{k=1}^{q^2} b_ky_k\ ,
\end{align}
as claimed.
\end{proof}
In the following Lemma, we show that $\inf_{(x,y) \in\Gamma(\cS) \cap \Omega} (\Vex_\Omega f)(x,y)$ can be expressed as an SDP. This provides
an efficient way of computing the lower bound~\eqref{eq:problemtooptimize2}.

\begin{lem}
  Let $\Omega\subset\mathbb{R}^{p^2}\times\mathbb{R}^{q^2}$ be a hyperrectangle and $\Gamma(\cS)$ be a set of vectors obtained from a set $\cS$ of semidefinite constraints as described in Lemma~\ref{lem:vectorproblem}. Furthermore, let $\Gamma(\cS) \cap \Omega$ be nonempty. Then the problem
  \begin{align}
    \inf_{(x,y)\in \Gamma(\cS)\cap \Omega} (\Vex_{\Omega}f)(x,y)
    \label{eq:sdplemma}
  \end{align}
  is a semidefinite program in $(X,Y,r)$, where $(X,Y)\in\selfadjointops(\mathbb{C}^p)\times\selfadjointops(\mathbb{C}^q)$ and $r\in\mathbb{R}^K$.
  \label{lem:sdpconstraints}
\end{lem}

\begin{proof}
Introduce the notation
\begin{align}
\hat{\ell}_j^0&=\ell_j\ ,\qquad &\hat{\ell}_j^1&=L_j\qquad &\textrm{for }j=1,\ldots,p^2\ ,\\
  \hat{m}_k^0&=m_k\ ,\qquad &\hat{m}_k^1&=M_k\qquad &\textrm{for }k=1,\ldots,q^2\ ,
\end{align}
for the lower- and upper bounds determining the hyperrectangle~$\Omega=\Omega(\ell,L,m,M)$. Then the functions $h_j^0, h_j^1$  introduced in Eq.~\eqref{eq:hfunctions} can be expressed as
\begin{align}
h^{b}_j(x_j,y_j)=\sigma_j\cdot (\hat{m}_j^b x_j+\hat{\ell}_j^b y_j-\hat{\ell}^b_j\hat{m}^b_j)\qquad\textrm{ for }\qquad b\in \{0,1\}\ .
\end{align}
We have by Lemma~\ref{lem:explicitconvexenve}
    \begin{align}\inf_{(x,y)\in \Gamma(\cS)\cap \Omega} (\Vex_{\Omega}f)(x,y)
&=\inf_{(x,y)\in\Gamma(\cS)\cap \Omega}
\sum_{j=1}^K
\max\{h_j^0(x_j,y_j),h_j^1(x_j,y_j)\} + \sum_{j=1}^{p^2} a_j x_j + \sum_{k=1}^{q^2} b_k y_k \\
&=\inf_{(x,y)\in\Gamma(\cS)\cap\Omega}
\inf_{\substack{r_1,\ldots,r_K\in\mathbb{R}\\
h_j^b(x_j,y_j)\leq r_j\\
\textrm{ for }j=1,\ldots,K,\ b \in \{0, 1\} }}
\sum_{j=1}^K r_j + \sum_{j=1}^{p^2} a_j x_j + \sum_{j=k}^{q^2} b_k y_k\ . \label{eq:tomaximizeexprv}
\end{align}
Here we have replaced each maximum
by a semidefinite program in a scalar, that is, we have used the identity
\begin{align}
\max\{h^0,h^1\}= \inf_{\substack{r\in\mathbb{R}\\
h^0\leq r\\
h^1\leq r}} r\qquad\textrm{ for all }h^0,h^1\in\mathbb{R}\ .
\end{align}

Let us first argue that
in Eq.~\eqref{eq:tomaximizeexprv}, we are optimizing over
a set of tuples $(x,y,r_1,\ldots,r_K)$ that can be described by SDP constraints.  Since $(x,y) \in \Gamma(\cS)$, $\Gamma$ is linear,  and $\cS$ is given by a set of semidefinite constraints on~$(X, Y)$, it suffices to verify that  the additional constraints imposed by $(x,y) \in \Omega$ and the constraints associated with the inner infimum in Eq.~\eqref{eq:tomaximizeexprv} can be expressed in semidefinite form. Indeed, with $(x,y)=\Gamma(X,Y)$ for $(X,Y)\in \Gamma^{-1}(\Omega)\cap \cS$, the constraints take the following form.  Since each function~$h^{b}_{j}(\cdot,\cdot)$ is affine-linear in both arguments,
the expression $h^b_j(x_j,y_j)$ is affine-linear in the operators~$X$ and $Y$.
Explicitly, we have
\begin{align}
h^{b}_{j}(x_j,y_j)&=\sigma_j \left[\tr\left(X\hat{m}_j^{b} \sum_{k=1}^{p^2} S_{k,j}\eta_k\right)+\tr\left(Y\hat{\ell}_j^{b}\sum_{\ell=1}^{q^2} T_{j,\ell}\xi_\ell\right)-\hat{\ell}^{b}_j\hat{m}^{b}_j \right]\
\end{align}
and the constraint
\begin{align}
h^b_j(x_j,y_j)\leq r_j
\end{align}
takes the form
\begin{align}
  \tr(XG^b_j)+\tr(YH^b_j)-r_j\leq s^b_j\ ,
  \label{eq:sdpconstraints1}
\end{align}
where
\begin{align}
  G^b_j&=\sigma_j \hat{m}_j^{b}\sum_{k=1}^{p^2} S_{k,j} \eta_k\ ,\qquad\qquad
  H^b_j=\sigma_j \hat{\ell}_j^{b}\sum_{\ell=1}^{q^2} T_{j,\ell} \xi_\ell\ ,\qquad\qquad
s^b_j=\sigma_j \hat{\ell}_j^b \hat{m}_j^b\
\end{align}
for each $j=1,\ldots,K$.
In addition, the constraints
\begin{align}
\ell_j\leq x_j&\leq L_j\qquad\textrm{ for }\qquad j=1,\ldots,p^2\qquad\textrm{and }\\
m_k\leq y_k&\leq M_k\qquad\textrm{ for }\qquad k=1,\ldots,q^2
\end{align}
become
\begin{align}
   \ell_j&\leq  \tr(X \sum_{k=1}^{p^2} S_{k,j}\eta_k)\leq L_j\qquad\textrm{ for }\qquad j=1,\ldots,p^2\qquad\textrm{and }\\
  m_k&\leq \tr(Y\sum_{\ell=1}^{q^2} T_{k,\ell}\xi_\ell)\leq M_k\qquad\textrm{ for }\qquad k=1,\ldots,q^2\ .
\label{eq:sdpconstraints2}
\end{align}
In summary, we are optimizing the objective function
\begin{align}
\sum_{j=1}^K r_j + \sum_{j=1}^{p^2} a_j x_j + \sum_{k=1}^{q^2} b_k y_k
\end{align}
over tuples $(X,Y,r)$ satisfying the constraints given by Eqs.~\eqref{eq:sdpconstraints1} and \eqref{eq:sdpconstraints2}.
Since this objective function is linear in $X$, $Y$, and $r$, respectively, the problem~\eqref{eq:sdplemma} is indeed a semidefinite program in $(X,Y,r)$.
\end{proof}
We again give pseudocode giving an algorithmic realization of Lemma~\ref{lem:sdpconstraints}, see subroutine \textsc{ComputeBoundsSDP} in Fig.~\ref{alg:computeboundssdp} of Appendix~\ref{app:pseudo}.

\subsubsection*{A branch-and-bound algorithm for jointly constrained semidefinite bilinear programs}
We are now ready to state our branch-and-bound algorithm which solves problem \eqref{eq:sdpproblem}. The algorithm closely follows the algorithm of Al-Khayyal and Falk and only the subroutines need to be adapted.

Our algorithm takes as input a set $\cS \subset \selfadjointops(\mathbb{C}^p) \times \selfadjointops(\mathbb{C}^q)$ defined by SDP constraints,  operators $Q \in \selfadjointops(\mathbb{C}^p \otimes \mathbb{C}^q), A \in \selfadjointops(\mathbb{C}^p), B \in \selfadjointops(\mathbb{C}^q)$ and a desired precision $\epsilon > 0$. It returns a value $\overline{\alpha}$ and an element $(X^*, Y^*) \in \cS$ such that (for the function $F$ defined in Eq.~\eqref{eq:sdpproblemfunction})
\begin{equation}
  \overline{\alpha} =  F(X^*,Y^*)  \leq \left(\inf_{(X,Y) \in \cS} F(X,Y)\right) + \epsilon\ .
\label{eq:outputalgnew}
\end{equation}

  The algorithm is given in Figs~\ref{alg:computevectorrep}--\ref{alg:branch-and-boundsdp} of Appendix~\ref{app:pseudo}. It follows exactly the same pattern as the branch-and-bound algorithm discussed in Section~\ref{sec:biconvex}, with the only modification that lower bounds on the objective function are computed by solving SDPs (instead of general convex programs).  In particular, with an identical analysis as  that of a general branch-and-bound algorithm (see~\cite{kahfalk83}), it follows that the iterative procedure described in the algorithm from Fig.~\ref{alg:branch-and-boundsdp} converges to a global solution of the problem~\eqref{eq:sdpproblem}. In other words, the terminating condition~\eqref{eq:outputalgnew} will always be reached. We note, however, that (as is typical for branch-and-bound algorithms), guarantees for the rate of convergence are typically not available.

\section{Application:  Dobrushin curves of quantum channels\label{sec:dobrushin}}
In this section, we apply our algorithm to
a problem in quantum information theory. We first 
explain this problem in Sections~\ref{sec:unconstrained} and~\ref{sec:constrained}, where we discuss the Dobrushin coefficient and the Dobrushin curve of a channel, respectively. In Section~\ref{sec:dobrushinasbilinear}, we  show that the problem of computing the Dobrushin curve is a semidefinite bilinear program. 
Finally, in Section~\ref{sec:numerics}, we present numerical results obtained by use of our algorithm.

\subsection{Converse to unconstrained coding over cascades\label{sec:unconstrained}}
Let $\Phi$ be a  channel.
Consider a setting where a message~$W$ is sent through a cascade
consisting of~$n$ copies of this channel, with a relay~$\cE_j$ applied before the $j$-th application of~$\Phi$.
We are interested in the amount of information the output 
\begin{align}
Y_n=\Phi\circ\cE_{n}\circ\cdots\circ\cE_2\circ\Phi\circ\cE_1(W)
\end{align}
of this cascade
provides about the input~$W$ for an optimal coding strategy (defined by the choice of relay channels~$\{\cE_j\}_{j=1}^n$). 
Denoting the output after applying the relay $\cE_j$ by $X_j$ and the output of the $j$-th channel~$\Phi$ by $Y_j$, we have the Markov property
\begin{equation}
  W  \rightarrow X_1 \rightarrow Y_1 \rightarrow X_2 \rightarrow Y_2 \rightarrow \cdots \rightarrow X_n \rightarrow Y_n\ .\label{eq:markovchainx}
\end{equation}
For the case where $\Phi=P_{Y|X}$ is a classical channel from
a set~$\cX$ to a set~$\cW$ and the message~$W$ is a binary random variable with uniform distribution on~$\{0,1\}$, a natural information measure is
the variational distance~$\|P_{Y_n|W=0}-P_{Y_n|W=1}\|_{1}$
between the output distributions for different inputs. Accordingly, a key quantity is
the Dobrushin coefficient
\begin{align}
\eta(P_{Y|X})&=\sup_{P_{X_0},P_{X_1}}
\frac{\|P_{Y|X}\circ P_{X_0}-P_{Y|X}\circ P_{X_1}\|_{1}}{\|P_{X_0}-P_{X_1}\|_{1}}\ ,\label{eq:dobrushincoefficientclassical}
\end{align}
where the optimization is over pairs~$(P_{X_0},P_{X_1})$ of distributions on~$\cX$ and where $P_{Y|X}\circ P_X$ is the distribution on~$\cY$ given by the push-forward of~$P_X$. Using the fact that $\|\cdot \|_{1}$ is non-increasing under application of channels, one can eliminate the choice of relays and conclude that
\begin{align}
\|P_{Y_n|W=0}-P_{Y_n|W=1}\|_{1}\leq 
\eta(P_{Y|X})^n\ 
\end{align}
independently of the coding strategy given by $\{\cE_j=P_{X_{j+1}|Y_j}\}$.

Similar reasoning applies to  quantum channels $\Phi:\cB(\cX)\rightarrow\cB(\cY)$ and relays $\cE_j:\cB(\cY)\rightarrow\cB(\cY)$ (i.e., completely positive trace-preserving maps) when a classical bit $W\in \{0,1\}$ is conveyed by encoding it into two states~$\rho_0,\rho_1$. The so-called Dobrushin coefficient
\begin{align}
\eta(\Phi)&=\sup_{\rho_0,\rho_1\in\cB(\cX)}
\frac{\|\Phi(\rho_0)-\Phi(\rho_1)\|_{1}}{\|\rho_0-\rho_1\|_{1}}\ ,
\end{align}
provides the upper bound
\begin{align}
\|\cF_n(\rho_0)-\cF_n(\rho_1)\|_1\leq \eta(\Phi)^n\ ,\qquad\textrm{ where }\qquad \cF_n=\Phi\circ\cE_{n}\circ\cdots\circ\cE_2\circ\Phi\circ\cE_1\ 
\end{align}
on the trace distance (defined by $\|A\|_1=\tr\sqrt{A A}$) between the output states. We refer to~\cite{hiairuskai} for a detailed discussion of the Dobrushin and other information measure based contraction coefficients for quantum channels.

One may ask how the maximum
output distinguishability
$\|\Phi(\rho_0)-\Phi(\rho_1)\|_1$
behaves as a function of the  distinguishability $\|\rho_0-\rho_1\|_1$ of the input states. The following lemma shows that this quantity is linear in~$\|\rho_0-\rho_1\|_1$, and thus not particularly exciting.
\begin{lem}\label{lem:dobrushinlinear}
Let $\delta\in [0,2]$. Then we have 
\begin{align}
\sup_{\substack{\rho_0,\rho_1\\
\|\rho_0-\rho_1\|\leq \delta
}} \|\Phi(\rho_0)-\Phi(\rho_1)\|_1&=(\delta/2)\cdot \eta(\Phi)\ .
\end{align}
\end{lem}
\begin{proof}
Consider the function $f:[0,2]\rightarrow [0,2]$ defined by 
\begin{align} 
f(\delta):=\sup_{\substack{\rho_0,\rho_1\\
\|\rho_0-\rho_1\|_1= \delta
}} \|\Phi(\rho_0)-\Phi(\rho_1)\|_1\ .
\end{align}
We first show that~$f$ is monotonically increasing. Indeed, suppose that $\delta\leq \delta'$, and let $\rho_0,\rho_1$ be states such that $\|\rho_0-\rho_1\|_1=\delta$ and 
$f(\delta)=\|\Phi(\rho_0)-\Phi(\rho_1)\|_1$.  Let 
$\rho_0-\rho_1=A_+-A_-$ be the decomposition of the difference into positive and negative parts
(i.e., $A_+\geq 0$ and $A_-\geq 0$). Then we have $\tr(A_+)=\tr(A_-)=\delta/2$. Accordingly, let us define the states $\sigma_{\pm}=\frac{2}{\delta}A_{\pm}$. Note that $\sigma_+$ and $\sigma_-$ are orthogonal by definition. Choose an arbitrary state $\sigma$ and define
\begin{align}
\rho'_0&=\frac{\delta'}{2}\sigma_++\left(1-\frac{\delta'}{2}\right) \sigma\\
\rho'_1&=\frac{\delta'}{2}\sigma_-+\left(1-\frac{\delta'}{2}\right) \sigma\ .
\end{align}
 Then it is easy to check (using the orthogonality of $\sigma_+$ and $\sigma_-$) that 
 \begin{align}
 \|\rho'_0-\rho'_1\|_1=\delta'\ .
 \end{align}
 Furthermore we have 
 \begin{align}
 \|\Phi(\rho'_0)-\Phi(\rho'_1)\|_1&=\frac{\delta'}{\delta} \|\Phi(\rho_0)-\Phi(\rho_1)\|_1\\
 &\geq \|\Phi(\rho_0)-\Phi(\rho_1)\|_1=f(\delta)\qquad\textrm{ for }\delta'\geq \delta\ .
 \end{align}
 This shows that $f(\delta')\geq f(\delta)$, as claimed. In particular, we also obtain $f(2)=\eta(\Phi)$.
 
 More generally, the above proof shows that 
 \begin{align}
   \frac{f(\delta')}{\delta'} \geq \frac{f(\delta)}{\delta}\qquad\textrm{ for all }\delta'\geq \delta\ ,
 \end{align}
and thus with $\delta'=2$
\begin{align}
f(\delta)\leq \frac{\delta}{2} \cdot \eta(\Phi)\qquad\textrm{ for all }\delta\in [0,2]\ .\label{eq:upperboundphif}
\end{align}
Now suppose that $\rho_0,\rho_1$ are states such that $\|\rho_0-\rho_1\|_1=2$ and $\eta(\Phi)=\|\Phi(\rho_0)-\Phi(\rho_1)\|_1$. Then $\rho_0$ and $\rho_1$ are orthogonal, implying that (again for an arbitrary state~$\sigma$) the states
\begin{align}
\rho'_0&=\frac{\delta}{2}\rho_0+\left(1-\frac{\delta}{2}\right)\sigma\\
\rho'_1&=\frac{\delta}{2}\rho_1+\left(1-\frac{\delta}{2}\right)\sigma
\end{align}
satisfy $\|\rho'_0-\rho'_1\|_1=\delta$. Since we also have
\begin{align}
\|\Phi(\rho'_0)-\Phi(\rho'_1)\|_1&=\frac{\delta}{2}\|\Phi(\rho_0)-\Phi(\rho_1)\|_1\\
&=\frac{\delta}{2}\eta(\Phi)\ 
\end{align}
we conclude that
\begin{align}
f(\delta)\geq (\delta/2)\cdot\eta(\Phi)\ .
\end{align}
With Eq.~\eqref{eq:upperboundphif} and the monotonicity of~$f$, the claim follows.
\end{proof}

\subsection{Converse to power-constrained coding over cascades\label{sec:constrained}}
Consider a modified cascade coding problem, where a power constraint is introduced for each of the inputs $X_j$ to the channel~$\Phi$ in~\eqref{eq:markovchainx}, for $j=1,\ldots,n$. In other words, each relay~$\cE_j$ is required to have power-constrained outputs. In the case where $\Phi=P_{Y|X}$ is a classical channel with continuous variable input $X$ (i.e., a random variable on $\mathbb{R}^m$), a natural power constraint is of the form 
\begin{align}
\mathop{\mathbb{E}}\left[\|X_j\|_2^2\right]\leq E\qquad\textrm{ for all }\qquad j\in \{1,\ldots,n\}\ ,\label{eq:energyconstraintg}
\end{align}
where $E>0$ is some constant (determining the available power) and $\|x\|_2^2=\sum_{k=1}^m x_k^2$ for $x\in\mathbb{R}^m$. Let $\cG_E$ be the set of distributions $P_X$ on $\mathbb{R}^m$ satisfying~\eqref{eq:energyconstraintg}. To analyze this scenario, 
Polyanskiy and Wu~\cite{PolyanskiyWu16} defined the function
\begin{align}
F_E(\delta) &=\sup_{\substack{P_{X_0},P_{X_1}\in \cG_E\\
\|P_{X_0}-P_{X_1}\|_1\leq \delta}}
\|P_{Y|X}\circ P_{X_0}-P_{Y|X}\circ P_{X_1}\|_{1}\qquad\textrm{ for }\qquad \delta\in [0,2]\ ,\label{eq:dobrushincurve}
\end{align}
which they call the Dobrushin curve of~$P_{Y|X}$. 
 Remarkably, Polyanskiy and Wu were able to compute~\eqref{eq:dobrushincurve} for the additive white Gaussian noise (AWGN) channel
 using a coupling argument. They then use this function to establish bounds on the distance~$\|P_{Y_n|W=0}-P_{Y_n|W=1}\|_{1}$: inductively applying Definition~\eqref{eq:dobrushincurve}, one obtains
 \begin{align}
 \|P_{Y_n|W=0}-P_{Y_n|W=1}\|_1 & \leq F^{\circ n}(2)\ ,
 \end{align}
 where $F^{\circ n}=F\circ\cdots\circ F$ is the $n$-fold composition of~$F$ (see Fig.~\ref{fig:dobrushincurveapplied} for an illustration). It should  be noted that the Dobrushin coefficient~\eqref{eq:dobrushincoefficientclassical}
is not meaningful for the AWGN channel: it evaluates to~$1$ and does not provide converse bounds.

\begin{figure}[h]
\begin{center}
  \includegraphics[width=0.7\linewidth]{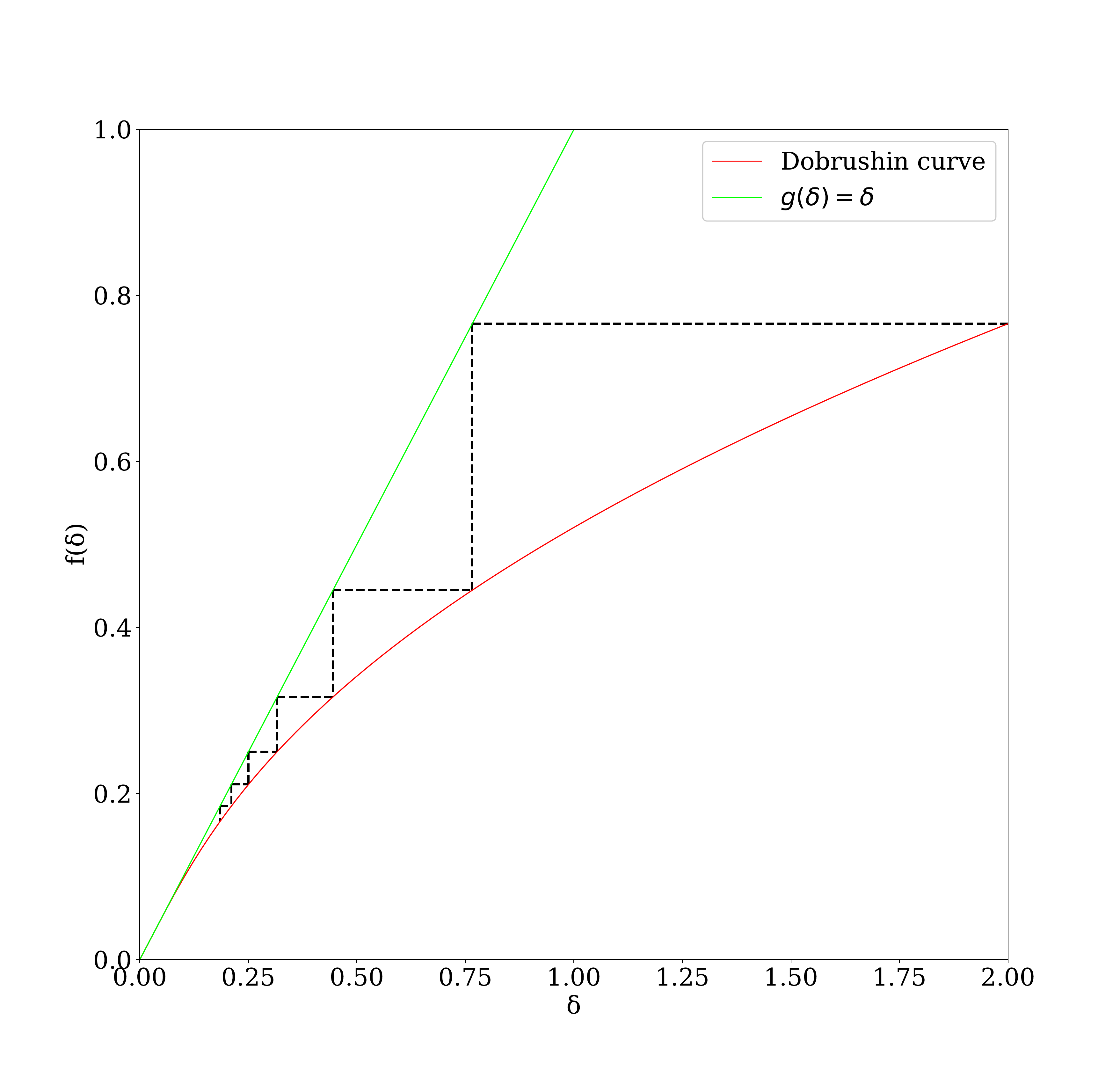}
\end{center}
\caption{Using the Dobrushin curve to obtain upper bounds on the information loss of a cascade of channels.\label{fig:dobrushincurveapplied}}
\end{figure}

Similar concepts are naturally defined for a quantum channel~$\Phi:\cB(\cX)\rightarrow\cB(\cY)$. In this case, a power constraint on states on $\cX$ can be defined by fixing a Hamiltonian $H$ (i.e., a self-adjoint operator) on $\cX$ and requiring that
the expected energy is less than a constant. For $E\in \mathbb{R}$,  let 
  \begin{align}
  \cG_E=\{\rho\in \cB(\cH)\ |\ \rho\geq 0, \tr \rho=1\textrm{ and } \tr(\rho H)\leq E\}\ 
  \end{align} 
  be the set of states satisfying this energy constraint. 
  We can then define\,---\,in analogy with~\eqref{eq:dobrushincurve}\,---\,the function
  \begin{align}
F_E(\delta) &=\sup_{\substack{\rho_0,\rho_1\in \cG_E\\
\|\rho_0-\rho_1\|_1\leq \delta}}
\|\Phi(\rho_0)-\Phi(\rho_1)\|_{1}\qquad\textrm{ for }\qquad \delta\in [0,2]\ .\label{eq:dobrushincurvequantum}
\end{align}
Contrary to the unconstrained case discussed in Lemma~\ref{lem:dobrushinlinear}, the function~\eqref{eq:dobrushincurvequantum} is not linear in~$\delta$, and its evaluation appears to be challenging in general.

\subsection{The Dobrushin curve as a semidefinite bilinear program\label{sec:dobrushinasbilinear}}
In this section, we show that the energy-constrained Dobrushin curve for 
 finite-dimensional quantum channels can be cast 
 as a semidefinite bilinear program of the form~\eqref{eq:sdpproblem}. This allows us to numerically compute the curve by applying our algorithm.

\begin{lem}
Consider a CPTPM $\Phi:\cB(\mathbb{C}^d)\rightarrow\cB(\mathbb{C}^d)$. Then we have
\begin{align}
  F_E(\delta)=\delta \max_{(P,Q,R,S)\in \Gamma(E,\delta)} \tr(P\Phi(R - S))\ ,\label{eq:FEoptimizationproblem}
\end{align}
where $\Gamma(E,\delta)$ is the set of quadruples $(P,Q,R,S)\in \cB(\mathbb{C}^d)^{\times 4}$ satisfying
\begin{align}
Q&\geq 0\ ,\;\tr(Q)=1\ ,\quad\textrm{ and } \tr(HQ)\leq E\ ,\label{eq:firstqhzero}\\
Q+\frac{\delta}{2}(R-S)&\geq 0\ ,\label{eq:nonnegativityQR}\\
\tr(H(Q+\frac{\delta}{2}(R-S))&\leq E\ ,\label{eq:energyboundlincomb}\\
0&\leq R\qquad\textrm{and }\qquad \tr(R)=1\ ,\label{eq:rcondition}\\
0&\leq S\qquad\textrm{and }\qquad \tr(S)=1\ ,\label{eq:scondition}\\
0&\leq P \leq I\ .\label{eq:pbounds}
\end{align}
\label{lem:dobrushin}
\end{lem}
To see more explicitly that the optimization problem~\eqref{eq:FEoptimizationproblem} 
is a semidefinite bilinear program, define the flip operator $\mathbb{F} = \sum_{i,j=1}^d \ket{ij} \bra{ji}$,
where $\{\ket{i}\}_{i=1}^d$ is an orthonormal basis of~$\mathbb{C}^d$. Using the identity $\tr(\mathbb{F}(A\otimes B))=\tr(AB)$ for all $A,B\in\cB(\mathbb{C}^d)$, we have
\begin{align}
\tr (P \Phi(R-S))&=\tr(\mathbb{F} (P\otimes \Phi(R)))-\tr(\mathbb{F} (P\otimes \Phi(S)))\nonumber\\
&=\tr\left((I\otimes\Phi^*)(\mathbb{F})(P\otimes R)\right)-\tr\left((I\otimes\Phi^*)(\mathbb{F})(P\otimes S)\right)\ ,\label{eq:linearprogramdobrex}
\end{align}
where $\Phi^*$ is the adjoint channel (with respect to the Hilbert-Schmidt inner product) i.e., it is defined by $\tr(A\Phi(B))=\tr(\Phi^*(A)B)$ for all $A,B\in\cB(\cH)$. This matches the form of~\eqref{eq:sdpproblem} with $X=P$, $Y=R\oplus S$, and 
$Q=\hat{Q}\oplus 0-Q_1\oplus 0\oplus Q_3$ where $Q_1\oplus Q_3=\hat{Q}$, and $\hat{Q}=(I\otimes \Phi^*)(\mathbb{F})$. 
\begin{proof}
For convenience, let $\Theta(E,\delta)$ denote the set of pairs $(\rho_0,\rho_1)$ of states satisfying
\begin{align}
\tr(H\rho_j)\leq E\qquad\textrm{ for }j=1,2\qquad\textrm{ and }\qquad \|\rho_0-\rho_1\|_1\leq \delta\ .\label{eq:Thetaedelta}
\end{align}
Suppose $(P,Q,R,S)\in \Gamma(E,\delta)$.
Set 
\begin{align}
\rho_1=Q\qquad\textrm{ and }\qquad \rho_0=Q+\frac{\delta}{2}(R-S)\ .
\end{align}
Because of \eqref{eq:firstqhzero}, $\rho_1$ is a state and satisfies the energy constraint, i.e., $\rho_1\in \cG_E$.
Similarly, $\rho_0$ is a state since it has unit trace because of~\eqref{eq:firstqhzero},~\eqref{eq:rcondition}, and~\eqref{eq:scondition}, and because it is non-negative by~\eqref{eq:nonnegativityQR}. By  Eq.~\eqref{eq:energyboundlincomb}, it also belongs to $\cG_E$.
Now observe that
\begin{align}
\|\rho_0-\rho_1\|_1&=\frac{\delta}{2}\|R-S\|_1\leq \delta\ ,
\end{align}
since both $R$ and $S$ are states (cf.~Eqs.~\eqref{eq:rcondition} and~\eqref{eq:scondition}) and hence $\|R-S\|_1\leq 2$. This shows that $(\rho_0,\rho_1)\in \Theta(E,\delta)$.  Furthermore, we have
\begin{align}
\delta\max_{0\leq P\leq I}\tr(P\Phi(R-S))&=\left\|\Phi\left(\frac{\delta}{2}(R-S)\right)\right\|_1=\|\Phi(\rho_0)-\Phi(\rho_1)\|_1\ .\label{eq:identityphirhozeroone}
\end{align}
We conclude that $F_E(\delta)\geq \delta\max_{(P,Q,R,S)\in \Gamma(E,\delta)}\tr(P\Phi(R-S))$.

To show the converse inequality, assume that $(\rho_0,\rho_1)\in \Theta(E,\delta)$. Then 
\begin{align}
\rho_0-\rho_1=A_+-A_-\label{eq:differenceexpansion}
\end{align}
for two orthogonal nonnegative operators $A_+,A_-$ satisfying
\begin{align}
\tr(A_+)=\tr(A_-)=\frac{\delta}{2}\ .
\end{align}
Set $Q = \rho_1$, $R = \frac{2}{\delta}A_+$, $S=\frac{2}{\delta}A_-$ and
\begin{align}
P&=\argmax_{0\leq P\leq I} \tr(P (\Phi(\rho_0)-\Phi(\rho_1))) \ .
\end{align}
Clearly, the quadruple $(P,Q,R,S)$ satisfies Eqs.~\eqref{eq:firstqhzero}, ~\eqref{eq:rcondition},~\eqref{eq:scondition} and \eqref{eq:pbounds}. 
It remains to check~\eqref{eq:nonnegativityQR}
and~\eqref{eq:energyboundlincomb}. Observe that by definition, we have  
\begin{align}
Q+\frac{\delta}{2}(R-S)&=\rho_1+A_+-A_- =\rho_0\ .
\end{align}
This implies that
\eqref{eq:nonnegativityQR}
and~\eqref{eq:energyboundlincomb} are also satisfied.

Since the identity~\eqref{eq:identityphirhozeroone} also holds by definition of
$(P,Q,R,S)$, we find 
\begin{align}
F_E(\delta) \leq \delta
\max_{(P,Q,R,S)\in\Gamma(E,\delta)}\tr(P\Phi(R-S)) \ .
\end{align}
 This concludes the proof.
\end{proof}
We note that the statement of Lemma~\ref{lem:dobrushin} simplifies somewhat in the case where the map $\Phi:\cB(\mathbb{C}^2)\rightarrow\cB(\mathbb{C}^2)$ is a qubit channel. This is because the operators
$A_{\pm}$ in Eq.~\eqref{eq:differenceexpansion} are orthogonal, and hence proportional to rank-$1$-projections $\proj{\varphi_{\pm}}$ which satisfy $\proj{\varphi_+}+\proj{\varphi_-}=I$. Here $I$ is the identity operator on $\mathbb{C}^2$. In particular, this means that we can eliminate $S=I-R$. Retracing the proof of Lemma~\ref{lem:dobrushin}, we obtain the following.

\begin{cor}
Consider a qubit channel  $\Phi:\cB(\mathbb{C}^2)\rightarrow\cB(\mathbb{C}^2)$. Then 
\begin{align}
  F_E(\delta)=\delta \max_{(P,Q,R)\in \Gamma(E,\delta)} \tr(P\Phi(2R-I))\ ,
\end{align}
where $\Gamma(E,\delta)$ is the set of triples $(P,Q,R)\in \cB(\mathbb{C}^2)^{\times 3}$ satisfying
\begin{align}
\tr(Q)&=1\ ,\; Q\geq 0\ ,\;\textrm{ and } \tr(HQ)\leq E\label{eq:firstqhzeroqub}\ ,\\
Q+\frac{\delta}{2}(2 R-I)&\geq 0\ ,\label{eq:nonnegativityQRqub}\\
\tr(H(Q+\frac{\delta}{2}(2R-I))&\leq E\ ,\label{eq:energyboundlincombqub}\\
\tr(R)&=1\ ,\label{eq:rnormalization}\\
0&\leq R\ ,\label{eq:Rlowerupperbound}\\
0&\leq P \leq I\ .\label{eq:pbounds2}
\end{align}
\label{cor:dobrushinqubit}
\end{cor}
One can furthermore add the condition $\tr(P) = 1$ as we know that the optimal $P$ is rank 1 and satisfies this condition.
Again we may recast this in the form of~\eqref{eq:sdpproblem} using the fact that (in analogy to~\eqref{eq:linearprogramdobrex})
\begin{align}
  \tr (P \Phi(2R-I))&=2\tr\left((I\otimes\Phi^*)(\mathbb{F})(P\otimes R)\right)- \tr\left(P \Phi(I)\right)\ .
\end{align}
In this case, we obtain both a bilinear term as well as a term which is linear in $P$.

\subsection{Numerical computation of Dobrushin curves\label{sec:numerics}}

According to Lemma~\ref{lem:dobrushin} (respectively Corollary~\ref{cor:dobrushinqubit}), we can use our biconvex programming algorithm to calculate Dobrushin curves for quantum channels.

For concreteness, we consider qubit channels. Let $\{ \sigma_j \}_{j=1}^3$ be the Pauli matrices
  \begin{align}
    \sigma_1 = \begin{pmatrix}0 & 1 \\ 1 & 0 \end{pmatrix},\qquad \sigma_2 = \begin{pmatrix}0 & -i \\ i & 0 \end{pmatrix},\qquad \sigma_3 = \begin{pmatrix}1 & 0 \\ 0 & -1 \end{pmatrix}\ .
  \end{align}
  A state $\rho$ (i.e. a non-negative operator with unit trace) can be represented as
  \begin{equation}
    \rho = \frac12 \left(I + \sum_{k=1}^3 w_k \sigma_k \right)\label{eq:embedding}
  \end{equation}
  with $w \in \mathbb{R}^3$ satisfying $\|w\|_2\leq 1$. The vector $w$ is called the Bloch vector of the state $\rho$.
  We remark that Eq.~\eqref{eq:embedding} provides an isometric identification of the set of states on~$\mathbb{C}^2$ with trace-norm, and the unit ball (with respect to the Euclidean norm) in~$\mathbb{R}^3$, see e.g., \cite[Chapter 9]{NielsenChuang}.
  
Without loss of generality, we will assume that the Hamiltonian under consideration is 
\begin{align}
H&=\sigma_3\ . \label{eq:HamiltonianH}
\end{align}  
In other words, we will be interested in states~$\rho$ having Bloch vectors $w=(w_1,w_2,w_3)$ satisfying an inequality of the form~$w_3\leq E$.
  
\subsubsection{Example: the dephasing channel\label{sec:dephasingchannel}}
As a first example consider a dephasing channel
$\Phi:\;\cB(\mathbb{C}^2) \rightarrow \cB(\mathbb{C}^2)$. For $a \in [0,1]$, this acts as
\begin{align}
  \Phi\left(\alpha_0 I + \sum_{k=1}^3 \alpha_k \sigma_k \right) &= \alpha_0 I + a\left(\alpha_1\sigma_1 + \alpha_2 \sigma_2\right) + \alpha_3 \sigma_3\qquad\textrm{ for all }\alpha\in\mathbb{R}^4\ .
\label{eq:dephasingdef}
\end{align}
The dephasing channel~\eqref{eq:dephasingdef} has the invariance property
\begin{align}
  \Phi\left( e^{i\theta \sigma_3} \rho e^{-i\theta \sigma_3} \right) = e^{i\theta \sigma_3} \Phi(\rho) e^{-i\theta\sigma_3}\qquad\textrm{ for all }\theta\in [0,2\pi) \textrm{ and }\rho\in\cB(\mathbb{C}^2)\  .\label{eq:invariancephi}
\end{align}
The Hamiltonian~\eqref{eq:HamiltonianH} is also invariant under rotations around the $\sigma_3$-axis, i.e., we have 
\begin{align}
H&=e^{-i\theta\sigma_3}He^{i\theta \sigma_3}\qquad\textrm{ for all }\theta\in [0,2\pi)\ .\label{eq:hamiltonianinva}
\end{align}
Eq.~\eqref{eq:hamiltonianinva}  implies that the set~$\cG_E$ of energy-constrained states 
is closed under the family of maps $\rho \mapsto e^{i\theta \sigma_3} \rho e^{-i\theta\sigma_3}$. By the invariance property~\eqref{eq:invariancephi} and the unitary invariance of the trace norm, we conclude that applying the joint rotations 
\begin{align}
(\rho_0,\rho_1) &\mapsto (e^{i\theta \sigma_3}\rho_0e^{-i\theta\sigma_3},e^{i\theta\sigma_3}\rho_1e^{-i\theta\sigma_3})\qquad\textrm{ for any }\theta\in [0,2\pi)
\end{align}
to a pair $(\rho_0,\rho_1)$ of states leaves their energies as well as the distances $\|\rho_0-\rho_1\|_1$ and  $\|\Phi(\rho_0)-\Phi(\rho_1)\|_1$ invariant. Because $\rho\mapsto e^{i\theta}\rho e^{-i\theta}$
amounts to the map
\begin{align}
(w_1,w_2,w_3) & \mapsto ((\cos 2\theta) w_1-(\sin 2\theta) w_2,(\sin 2\theta) w_1+(\cos 2\theta)w_2,w_3)
\end{align}
on the level of Bloch vectors $w=(w_1,w_2,w_3)\in\mathbb{R}^3$, we conclude the following: for any fixed energy~$E$, there is a pair of states $(\rho_0,\rho_1)\in \Theta(E,\delta)$ (see Eq.~\eqref{eq:Thetaedelta}) such that 
\begin{align}
F_E(\delta)&=\|\Phi(\rho_0)-\Phi(\rho_1)\|_1\ ,
\end{align}
(i.e., the states achieve the optimum in the definition of the Dobrushin curve), and such that the Bloch vector~$w$ of $\rho_1$ lies in the $(\sigma_1,\sigma_3)$-plane, i.e., 
\begin{align}
w_2&=0\ .\label{eq:wtwocondition}
\end{align}
In the semidefinite bilinear program introduced in Corollary~\ref{cor:dobrushinqubit} (where $Q$ corresponds to~$\rho_1$), this means that we may add the constraint
\begin{align}
\tr(\sigma_2 Q) = 0\label{eq:dobrushinadditionalconstraint}
\end{align}
without changing the value of the optimization problem.

In Figs.~\ref{fig:dobrushindephasing} and~\ref{fig:dobrushincurveapplieddephasing}, we show numerically computed Dobrushin curves for dephasing channels. These are applied by using the formulation as a semidefinite bilinear program (see Corollary~\ref{cor:dobrushinqubit}), and imposing the constraint~\eqref{eq:dobrushinadditionalconstraint}. 

  \begin{figure}[h]
    \subfloat[The Dobrushin curve for the dephasing channel $\Phi$ for $a=0.3$ The algorithm required an average of $88$ elements (worst case: $178$ elements) in the partition $\cP$ to reach the desired precision.][The Dobrushin curve for the dephasing channel\\ $\Phi$ for $a=0.3$. The algorithm required an average of\\ $88$ elements (worst case: $178$ elements) in the partition\\ $\cP$ to reach the desired precision.\label{fig:dobrushincurve}]
  {\includegraphics[width=0.5\linewidth]{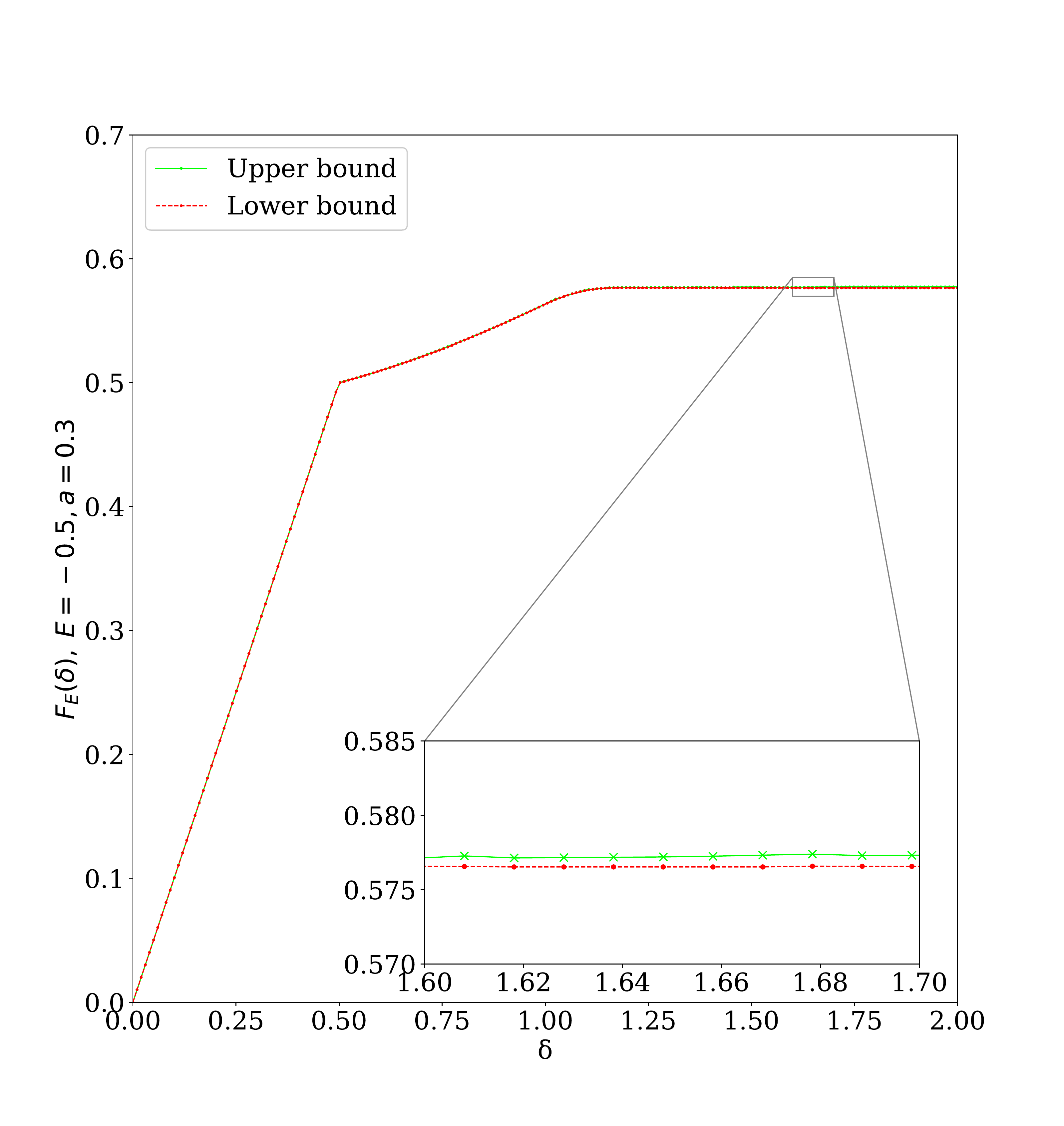}}
  \subfloat[The Dobrushin curve for the dephasing channel $\Phi$ for $a=0.5$. The algorithm required an average of $92$ elements (worst case: $142$ elements) in the partition $\cP$ to reach the desired precision.\label{fig:dobrushincurvetwo}]{\includegraphics[width=0.5\linewidth]{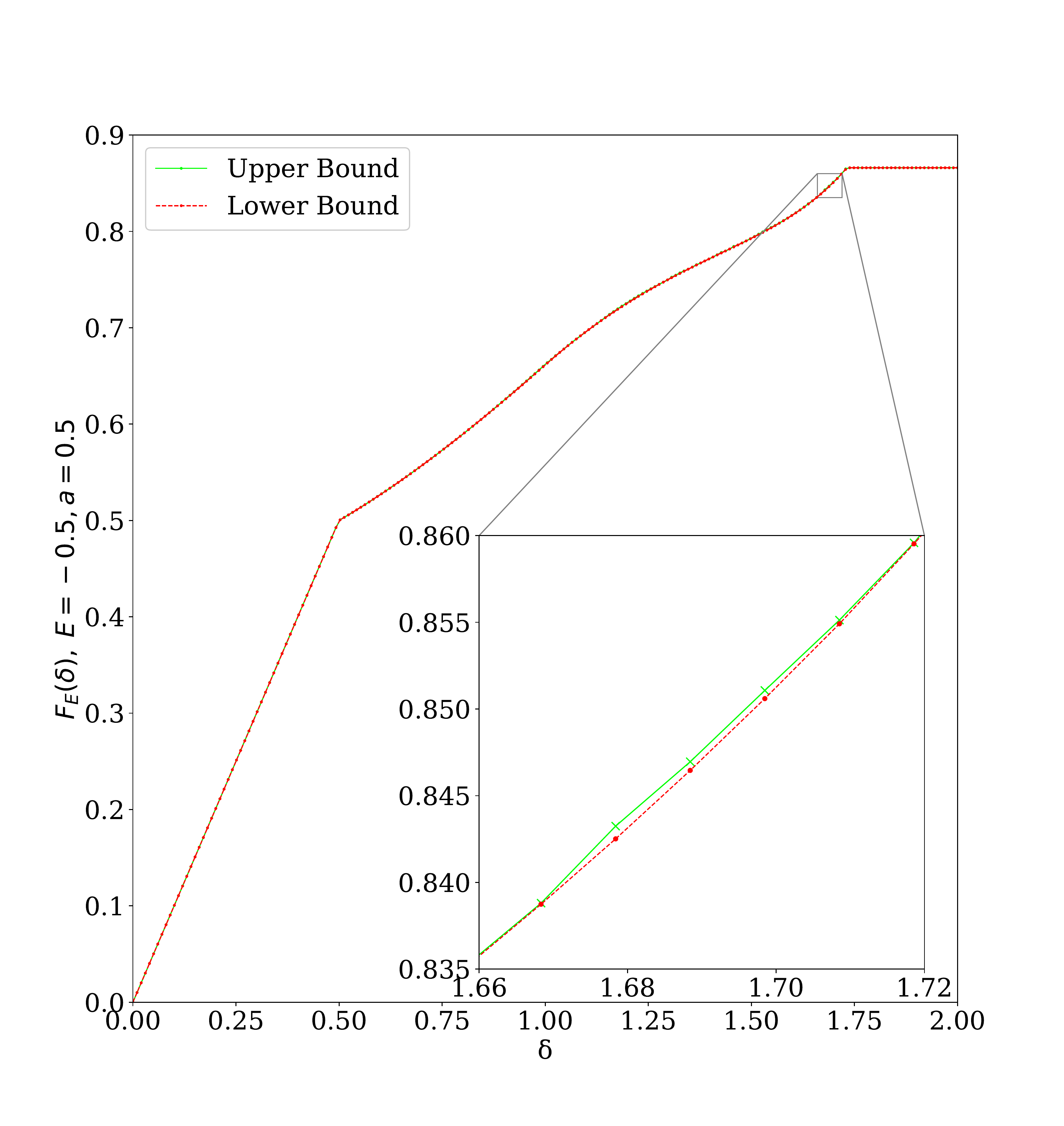}} 
  \caption{The Dobrushin curves of the dephasing channels $\Phi$ for two different values of $a$, for $E = -0.5$. For $200$ values of $\delta \in [0,2]$, the algorithm was run with a desired precision of $\epsilon = 10^{-3}$.\label{fig:dobrushindephasing}}
\end{figure}

\begin{figure}[t]
  \subfloat[The procedure to obtain upper bounds on the output distinguishability of a cascade of dephasing channels.][The procedure to obtain upper bounds on the\\ output distinguishability of a cascade of channels. \label{fig:dobrushin_loss_procedure_dephasing}]
  {\includegraphics[width=0.5\linewidth]{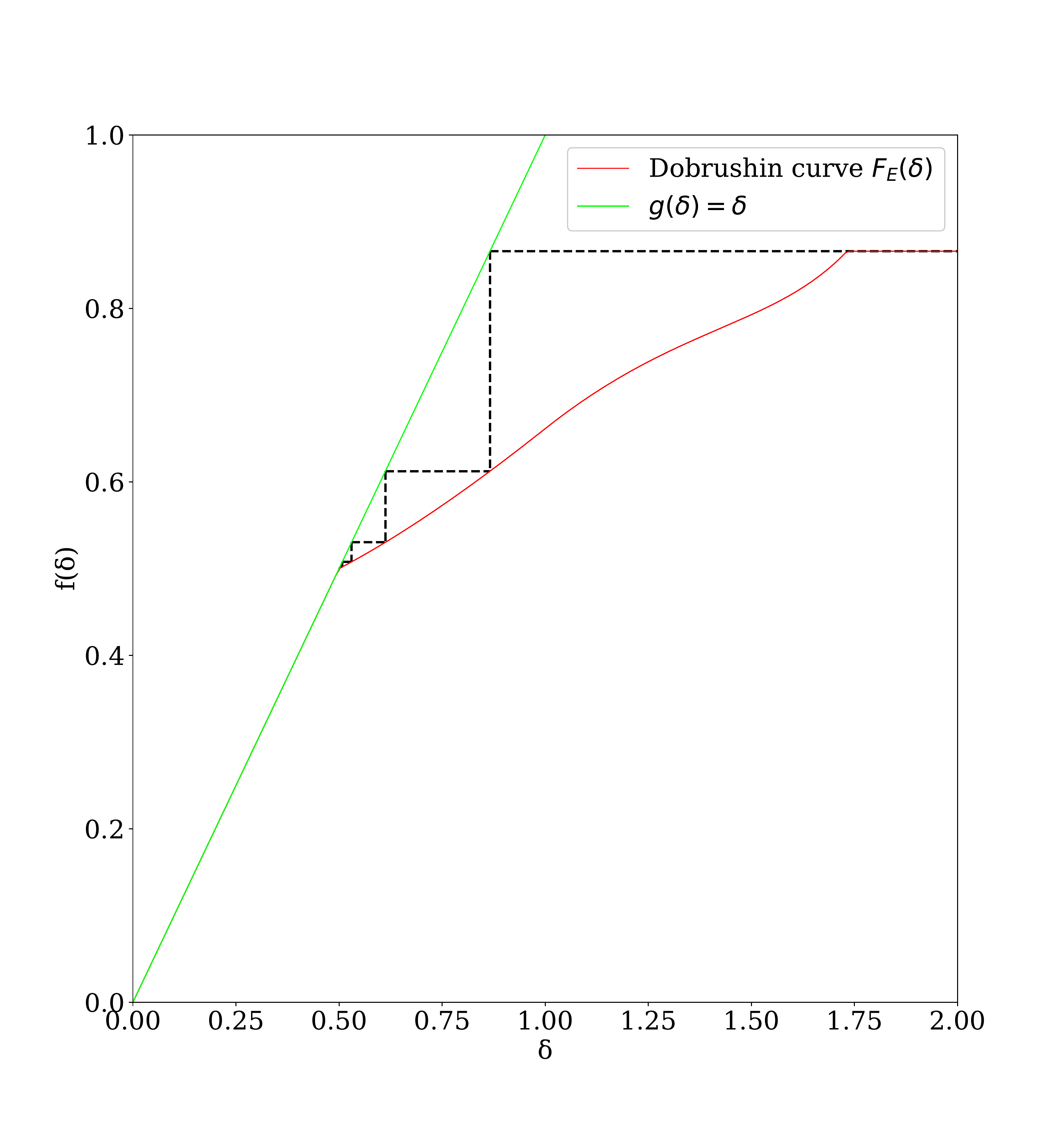}}
  \subfloat[The maximal output distinguishability for a cascade of phase-damping channels ($E = -0.5$), as a function of the length of the cascade.\label{fig:dobrushin_loss}]
  {\includegraphics[width=0.5\linewidth]{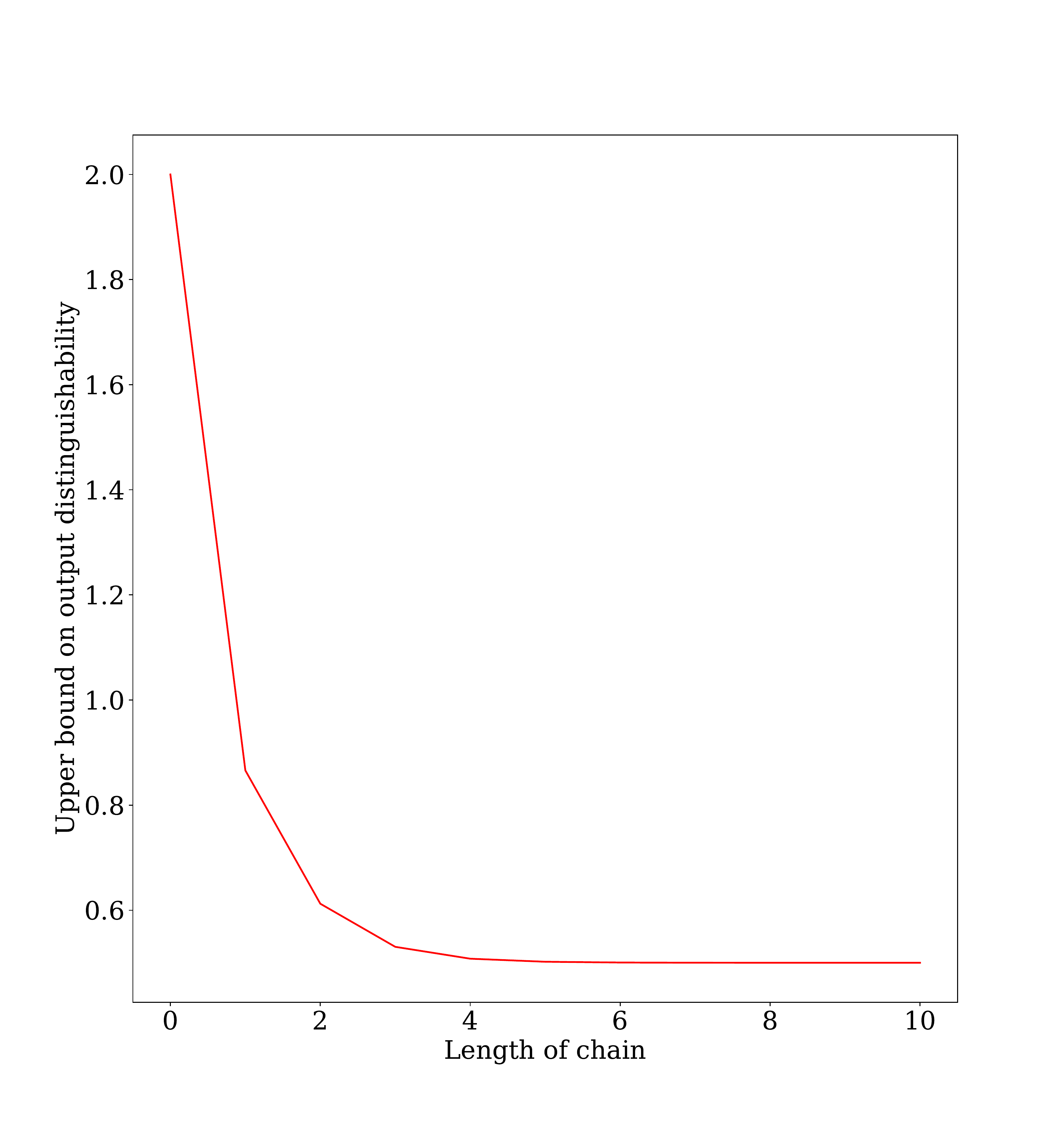}}
  \caption{Using the Dobrushin curve to obtain upper bounds on the information loss of a cascade of dephasing channels $\Phi$, for $E = -0.5$ and $a = 0.5$.\label{fig:dobrushincurveapplieddephasing}}
\end{figure}

In Fig.~\ref{fig:dobrushincurveconvergence}, we present numerical data  illustrating the importance of exploiting continuous symmetries by a constraint as in Eq.~\eqref{eq:dobrushinadditionalconstraint}. It is well-known that branch-and-bound algorithms perform badly in the context of such symmetries, hence it is important to include such constraints. Note that more generic channels as discussed in Section~\ref{sec:genericqubitdobrushin} typically do not exhibit such continuous symmetries. 
\begin{figure}
  \subfloat[Convergence of the algorithm with the symmetry constraint~\eqref{eq:dobrushinadditionalconstraint}][Convergence of the algorithm with the symmetry\\ constraint~\eqref{eq:dobrushinadditionalconstraint}\label{fig:dobrushincurveconvergencewith}]{\includegraphics[width=0.5\linewidth]{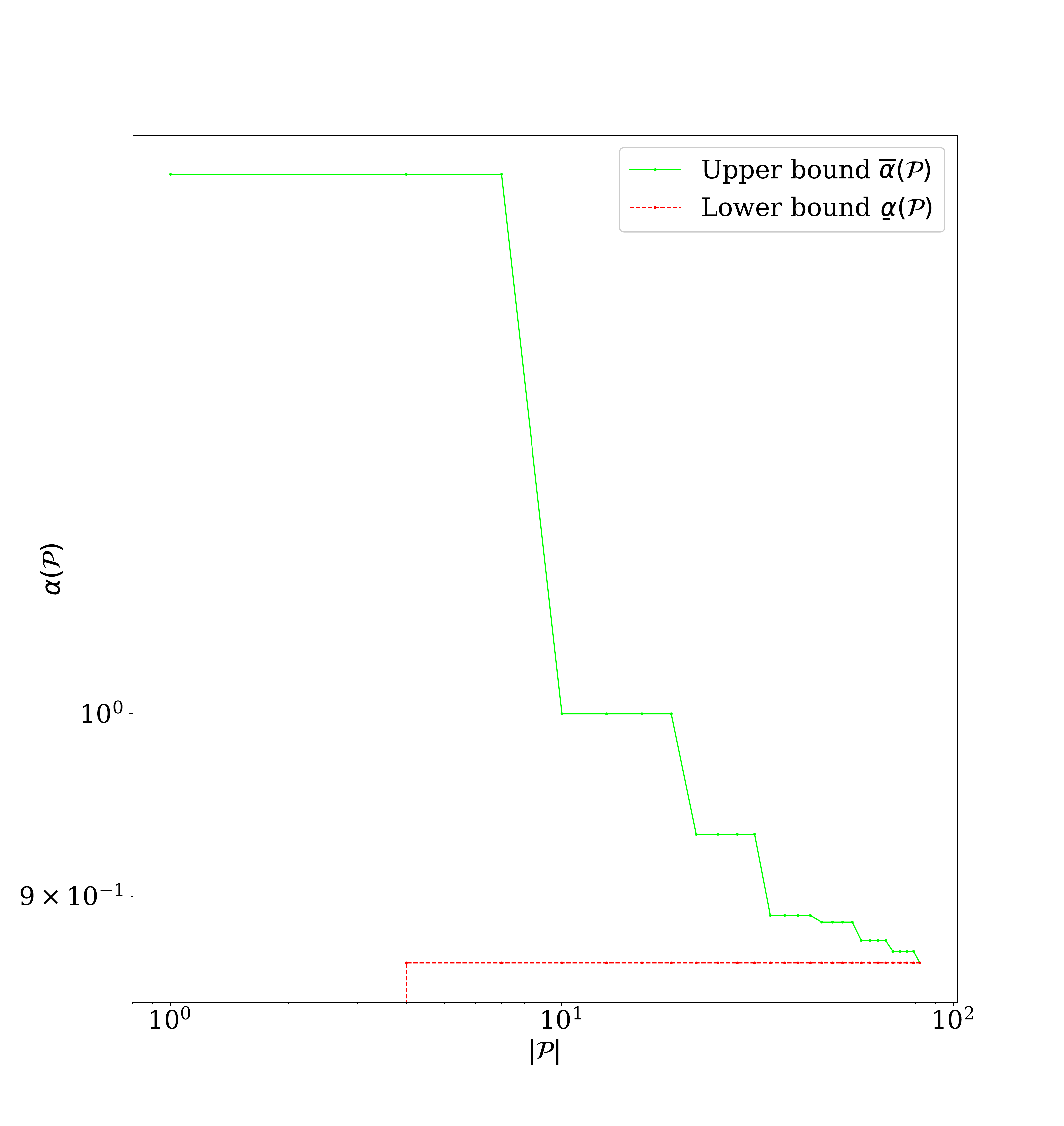}}
  \subfloat[Convergence of the algorithm without the symmetry constraint~\eqref{eq:dobrushinadditionalconstraint} \label{fig:dobrushincurveconvergencewithout}]{\includegraphics[width=0.5\linewidth]{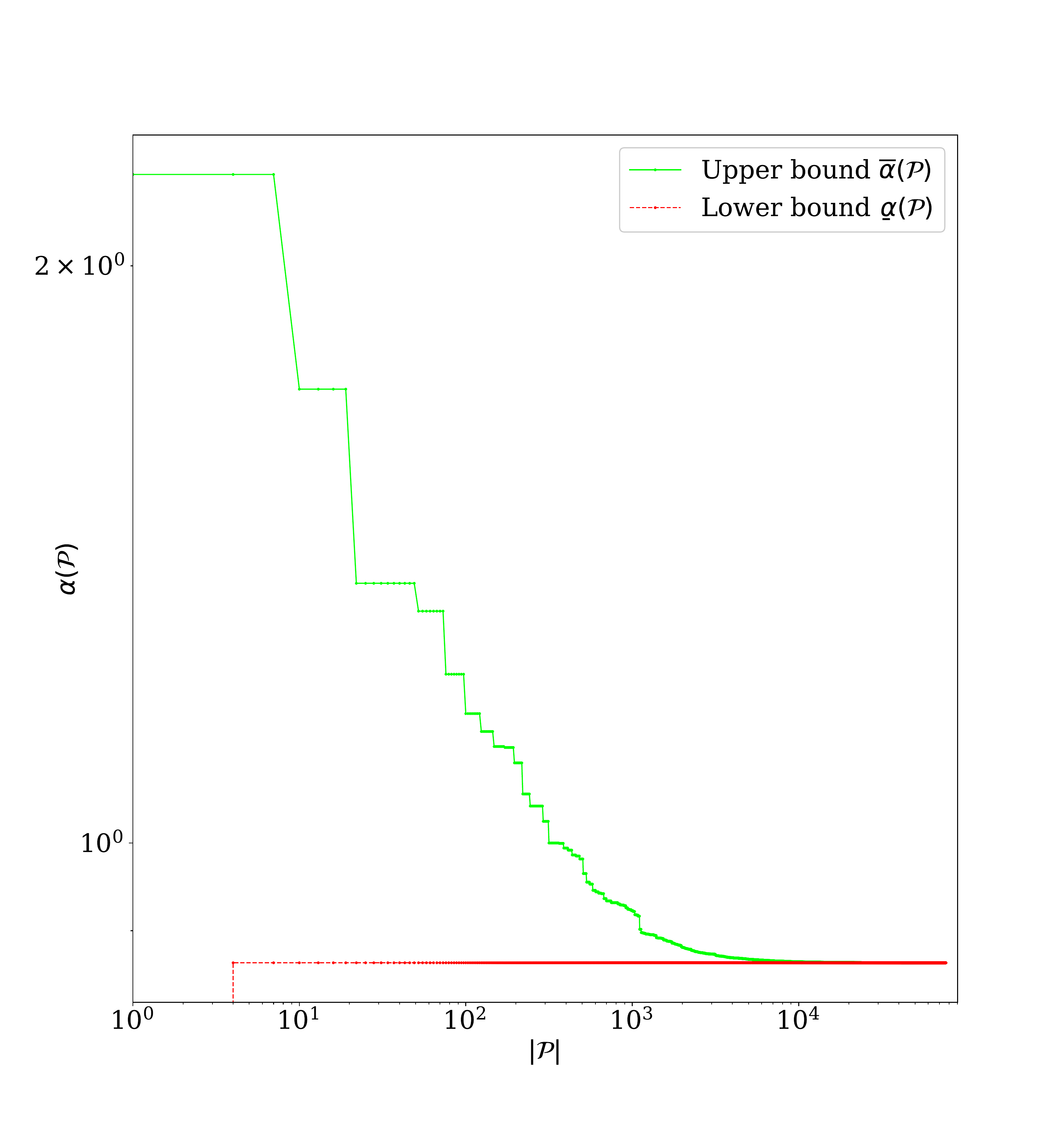}}
  \caption{Upper and lower bounds $\overline{\alpha}(\cP)$, $\underline{\alpha}(\cP)$ as a function of the size~$|\cP|$ of the partition. This provides a measure for the speed of convergence of the algorithm. The figures are for $E=-0.5$, $a=0.5$, $\delta = 2$, and $\epsilon = 10^{-5}$.
  \label{fig:dobrushincurveconvergence}}
\end{figure}

\subsubsection{Discussion of the dephasing channel}
We note that our algorithm also provides\,---\,in addition to the value $F_E(\delta)$ of the Dobrushin curve\,---\,a pair of states $(\rho_0,\rho_1)\in \Theta(E,\delta)$ satisfying $F_E(\delta)=\|\Phi(\rho_0)-\Phi(\rho_1)\|_1$. We call such a pair of states optimal for the Dobrushin curve. In the special case of the dephasing channel defined in Eq.~\eqref{eq:dephasingdef} for some $a \in [0,1]$, we can  provide the following description of such pairs (valid for instance for Fig.~\ref{fig:dobrushincurvetwo}, i.e., for $a=0.5$ and $E=-0.5$). We note that this description is  based on a heuristic geometric analysis of the problem. However, our numerical data shows that the following pairs of states are indeed optimal. We also note that\,---\,while a full analytical proof of optimality may in principle be constructed for the dephasing channel, such a brute-force calculation is unlikely to be achievable e.g., for generic qubit or qutrit channels, where symmetry arguments are not applicable and positivity constraints are particularly difficult to deal with.  

 Recall from Eq.~\eqref{eq:wtwocondition} that we can assume without loss of generality that one of the states~$(\rho_0,\rho_1)$\,---\,say $\rho_0$ for concreteness\,---\,has Bloch vector lying in the plane orthogonal to~$(0,1,0)$. It turns out that~$\rho_1$ can also be chosen to lie in this plane. 
 In~Fig.~\ref{fig:dobrushinoptimalcodestates}, we show the projection of
 the Bloch sphere onto this plane, and illustrate a choice of optimal code states (in terms of their Bloch vectors). 
\begin{figure}
\begin{center}
  \subfloat[The regime $0\leq \delta\leq 1-|E|$.  
   \label{fig:computingblochdephasingzero}]
  {\includegraphics[width=0.32\linewidth]{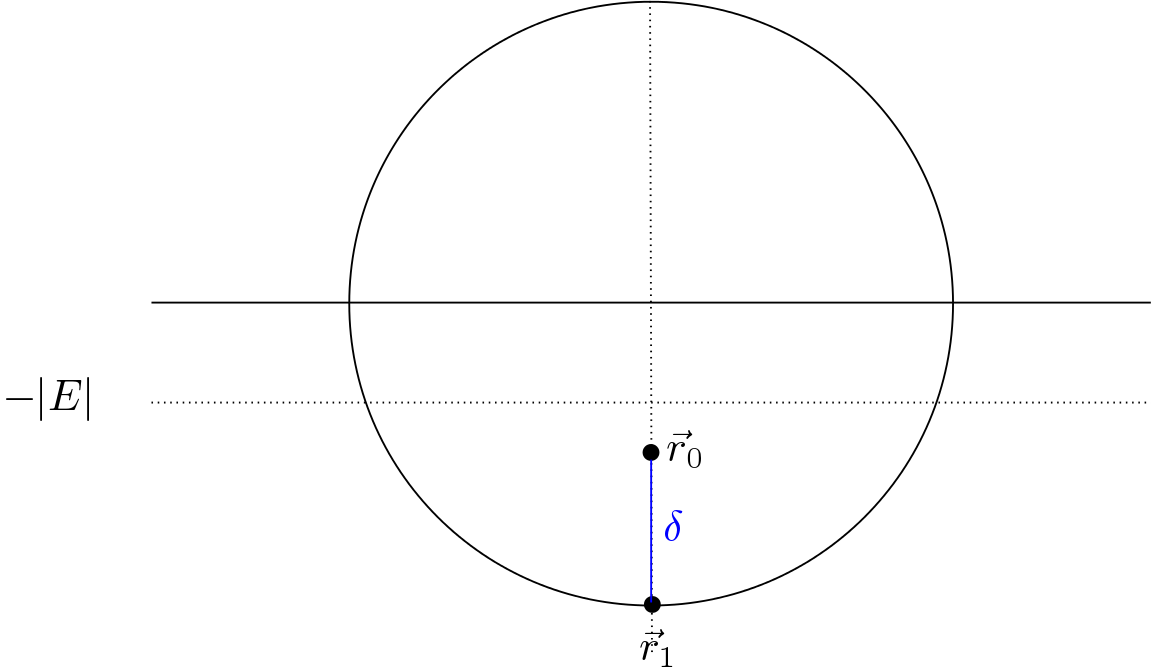}}
  \subfloat[The regime $1-|E| \leq \delta \leq \sqrt{2(1-|E|)}$.][The regime\\ $1-|E| \leq \delta \leq \sqrt{2(1-|E|)}$.\label{fig:computingblochdephasingone}]
  {\includegraphics[width=0.32\linewidth]{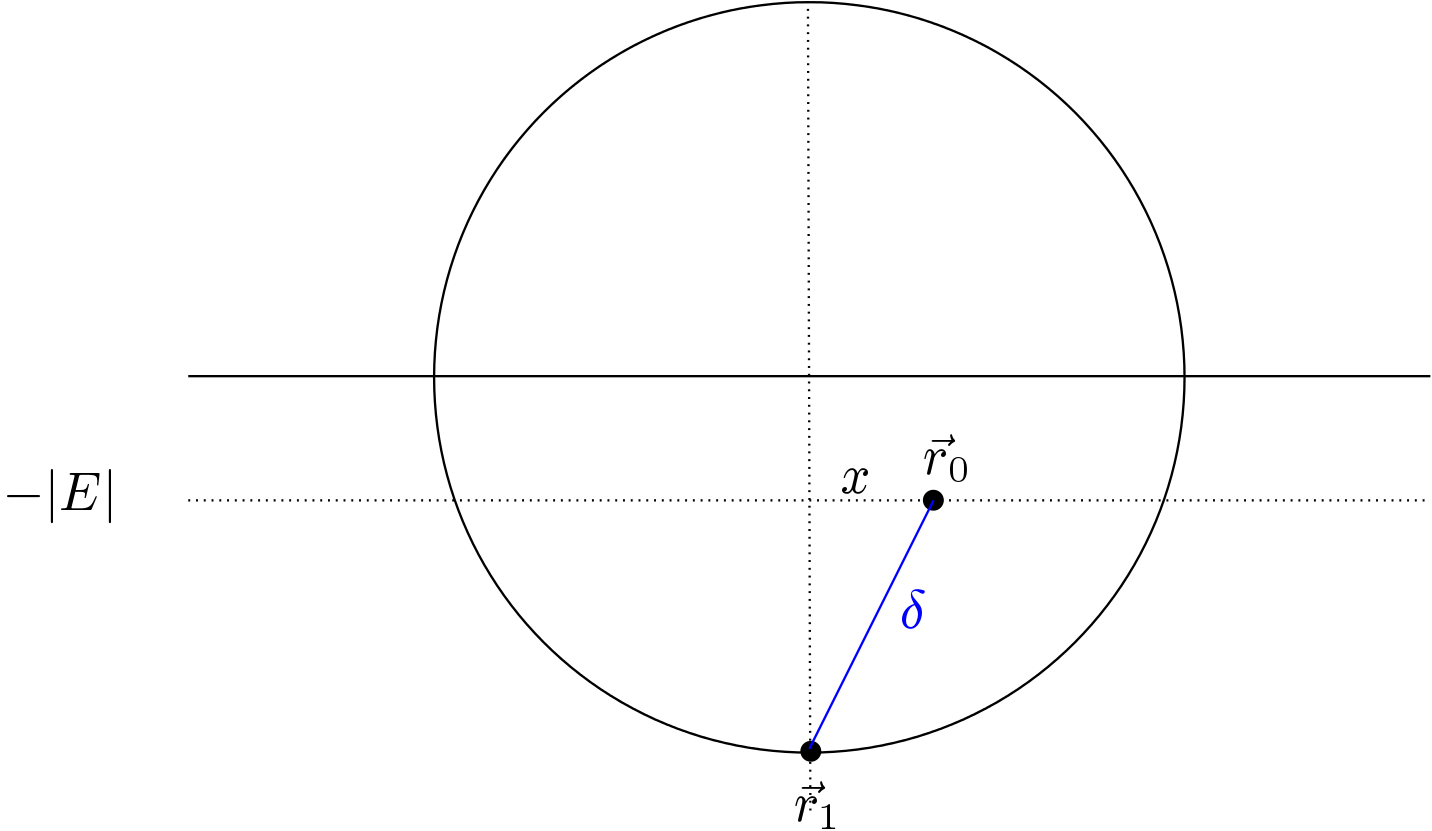}}
  \subfloat[The regime $\sqrt{2(1-|E|)} \leq \delta \leq 2\sqrt{1-|E|^2}$][The regime\\ $\sqrt{2(1-|E|)} \leq \delta \leq 2\sqrt{1-|E|^2}$.\label{fig:computingblochdephasingtwo}]
    {\includegraphics[width=0.32\linewidth]{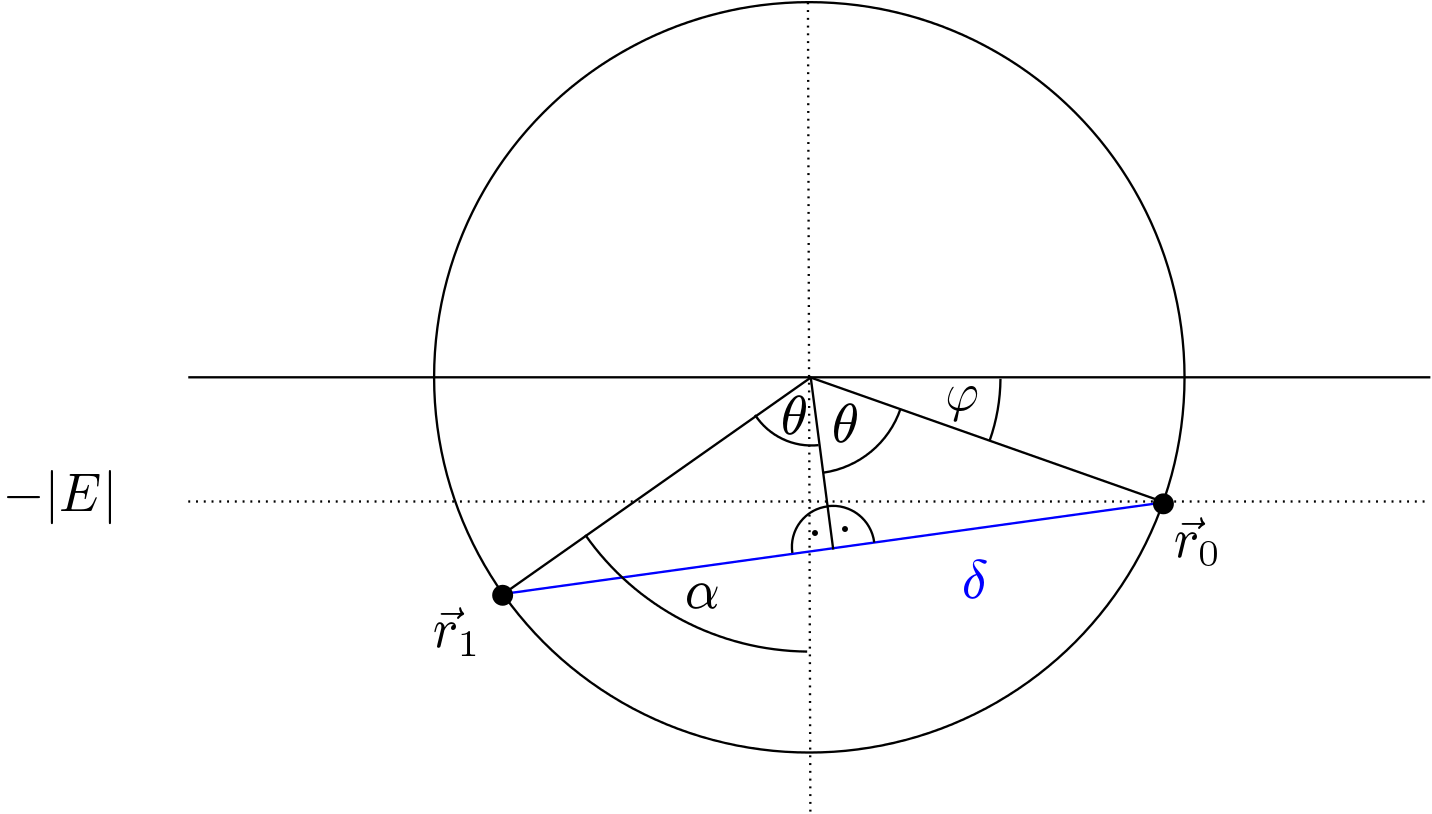}}
    \end{center}
    \caption{Bloch vectors $\vec{r}_0$ and $\vec{r}_1$ of a pair of optimal  states for the Dobrushin curve of the dephasing channel.\label{fig:dobrushinoptimalcodestates}}
\end{figure}
More precisely, we identify three regimes:
\begin{description}
\item[Regime I:] for $\delta\in [0,1-|E|]$, an optimal pair $(\rho_0,\rho_1)=\Theta(E,\delta)$
  is given by the pure state $\rho_1~=~\proj{1}$ with Bloch vector~$\vec{r}_1=(0,0,1)$ and a mixed state~$\rho_0$ whose Bloch vector~$\vec{r}=(0,0,1-\delta)$ also lies on the $e_3$-axis.see Fig.~\ref{fig:computingblochdephasingone}
\item[Regime II:] for $1-|E|\leq \sqrt{2(1-|E|)}$, 
an optimal pair is given by $\rho_1=\proj{1}$ as in Regime II and $\rho_0$ 
having Bloch vector~$(x,0,-|E|)$, with $x$ chosen such that $\|\rho_0-\rho_1\|_1=\delta$,
see Fig.~\ref{fig:computingblochdephasingtwo}.
\item[Regime III:] for $\sqrt{2(1-|E|)}\leq \delta\leq 2\sqrt{1-|E|^2}$,
we can choose the state~$\rho_0$ to have Bloch vector given by the ``eastern'' point of intersection of the projection of the Bloch sphere onto the plane orthogonal to $(0,1,0)$, and the plane $(x,y,-|E|)$. On the other hand, $\rho_1$ a pure state at distance~$\delta$, see Fig.~\ref{fig:dobrushinoptimalcodestates}.
\end{description}
This completes the description, as there are no pairs of states at distance~$\delta>2\sqrt{1-|E|^2}$ which belong to the energy-constrained subset~$\cG_E$. We now analyze this ``coding'' strategy for the dephasing channel and show the following:

\begin{lem}
We have $(\rho_0,\rho_1)\in\Theta(E,\delta)$ in all three regimes. In particular, these pairs of states give the following lower bound on the Dobrushin curve of the dephasing channel:
\begin{align}
f_E(\delta)&\geq \begin{cases}
\delta \qquad &\textrm{ if } \delta<1-|E|\\
g_E(\delta) &\textrm{ if } 1-|E|\leq \delta\leq \sqrt{2(1-|E|)}\\
h_E(\delta)&\textrm{ if } \sqrt{2(1-|E|)}\leq \delta\leq 2\sqrt{1-|E|^2}\\
2a \sqrt{1-|E|^2} &\textrm{ if }\delta>2\sqrt{1-|E|^2}\ 
\end{cases}\label{eq:lowerboundx}
\end{align}
where
\begin{align}
g_E(\delta)&:=\sqrt{a^2(\delta^2-(1-|E|)^2)+(1-|E|)^2}\\
h_E(\delta)&:=\Bigl[\left(|E|+\cos\left(2\arccos(\delta/2)+\arccos(|E|)\right)\right)^2\\
&\quad +a^2\left(\sqrt{1-|E|^2}+\sin\left(2\arccos(\delta/2)+\arccos(|E|)\right)\right)^2\Bigr]^{1/2}\ .
\end{align}
\end{lem}
The proof relies on elementary geometry. 
The curve given by the rhs.~of Eq.~\eqref{eq:lowerboundx} matches the numerically observed Dobrushin curve shown in Fig.~\ref{fig:dobrushincurvetwo}; this shows that the pairs of states considered above are indeed optimal. 
\begin{proof}
We consider each regime separately.
\begin{description}
\item[Regime I:] 
Consider $\delta\in [0,1-|E|]$. In this case, the choice
\begin{align}
\vec{r}_1&=(0,0,-1+|E|)\ ,\\
\vec{r}_2&=(0,0,-1)
\end{align}
is optimal and leads to $\|\vec{r}'_1-\vec{r}'_2\|_1=\|\vec{r}_1-\vec{r}_2\|_1$ for the output Bloch vectors $\vec{r}'_1$ and $\vec{r}'_2$. 

\item[Regime II:]

Now consider $\delta\in [1-|E|,\sqrt{2(1-|E|)}]$ see  Fig.~\ref{fig:dobrushinoptimalcodestates}.
The initial Bloch vectors are 
\begin{align}
\vec{r}_1&=(x,0,-|E|)\ ,\\
\vec{r}_2&=(0,0,-1)\ ,
\end{align}
where $x^2+(1-|E|)^2=\delta^2$, i.e.,
$x=\sqrt{\delta^2-(1-|E|)^2}$.
The Bloch vectors after application of the channel are 
\begin{align}
\vec{r}'_1&=(ax,0,-|E|)\ ,\\
\vec{r}'_2&=(0,0,-1)\ ,
\end{align}
such that 
\begin{align}
\|\vec{r}'_1-\vec{r}'_2\|_1&=\sqrt{a^2x^2+(1-|E|)^2}\\
&=\sqrt{a^2(\delta^2-(1-|E|)^2)+(1-|E|)^2}\ ,
\end{align}
for $\delta\in [1-|E|,\sqrt{(1-|E|)^2+1}]$.

\item[Regime III:]
Let us now look at $\delta\in [\sqrt{2(1-|E|)},2\sqrt{1-|E|^2}]$. 
Consider Fig.~\ref{fig:dobrushinoptimalcodestates}.
Assume that $\rho_1$ has Bloch vector given by
\begin{align}
\vec{r}_1&=(\cos\varphi,0,-\sin\varphi)\ ,
\end{align}
where $\sin\varphi=|E|$, and that $\rho_2$ has Bloch vector~$\vec{r}_2$ specified by an angle~$\theta$ as in Fig.~\ref{fig:computingblochdephasingtwo}.  The figure shows that
$\delta=2\sin \theta$
and $\alpha=2\theta+\varphi-\pi/2$. 
Thus, the Bloch vector~$\vec{r}_2$ is 
\begin{align}
\vec{r}_2&=(-\sin \alpha,0,-\cos\alpha)\\
&=(-\sin (2\theta+\varphi-\pi/2), 0,-\cos(2\theta+\varphi-\pi/2))\ .
\end{align}
These get  mapped to
\begin{align}
\vec{r}_1'&=(a\cos\varphi,0,-\sin\varphi)\ ,\\
\vec{r}_2'&=(-a\sin (2\theta+\varphi-\pi/2), 0,-\cos(2\theta+\varphi-\pi/2))\ ,
\end{align}
such that
\begin{align}
\|\vec{r}_1'-\vec{r}_2'\|_1&=\bigl(a^2\left(\cos\varphi+\sin(2\theta+\varphi-\pi/2)\right)^2\\
&\qquad +\left(\sin\varphi-\cos(2\theta+\varphi-\pi/2)\right)^2\bigr)^{1/2}\ ,
\end{align}
where $\theta=\arcsin(\delta/2)$ and $\varphi=\arcsin |E|$. 
\end{description}

\begin{figure}[t]
  \subfloat[The Dobrushin curve for the dephasing channel $\Phi_{\pi/9}$ with rotated principal axes. The algorithm required an average of $898$ elements (worst case: $2056$ elements) in the partition $\cP$ to reach the desired precision.][The Dobrushin curve for the dephasing channel\\ $\Phi_{\pi/9}$ with rotated principal axes. The algorithm\\ required an average of $898$ elements (worst case: $2056$\\ elements) in the partition $\cP$ to reach the desired precision. \label{fig:dobrushinrotated1}]
{\includegraphics[width=0.5\linewidth]{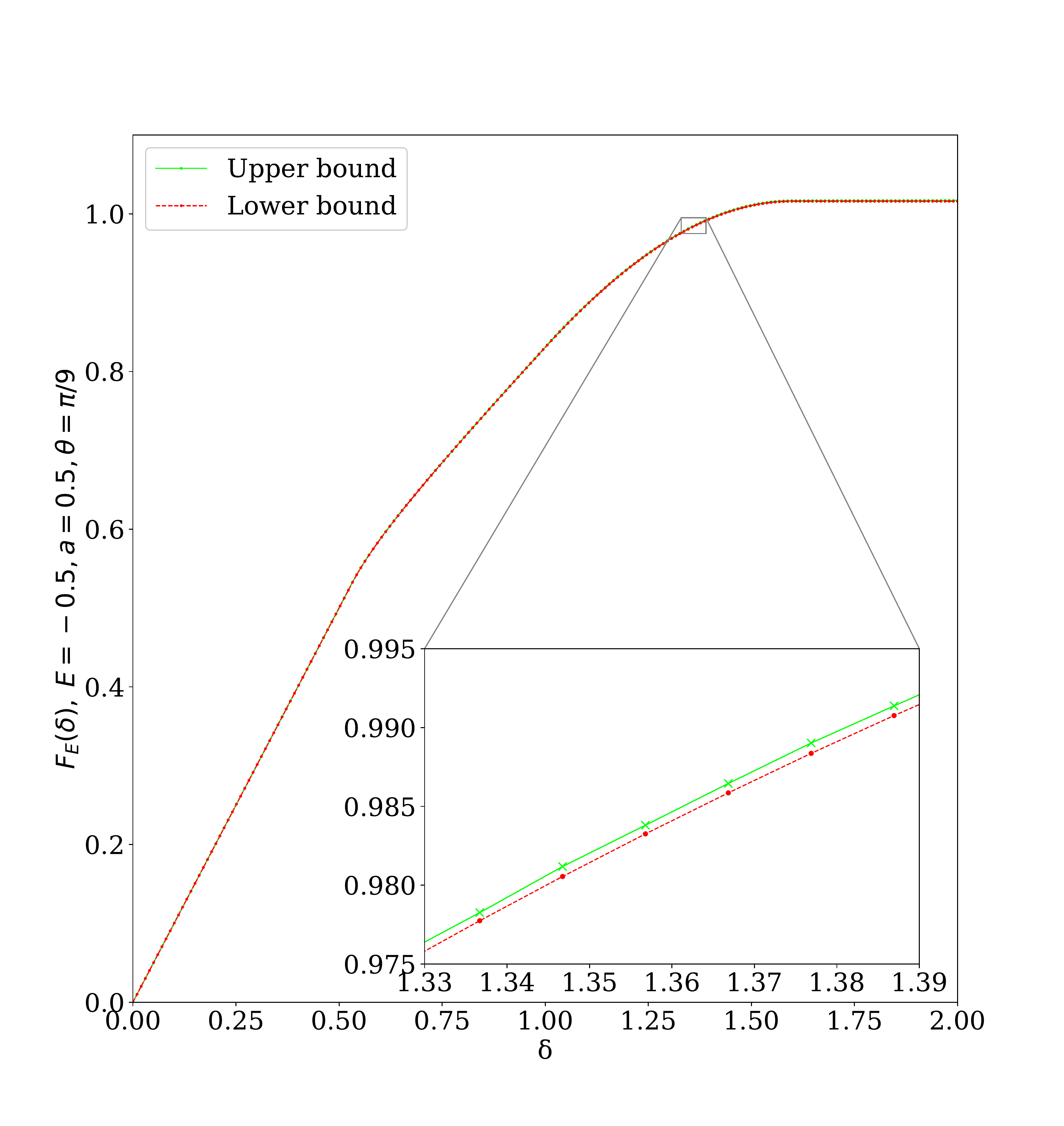}}
\subfloat[The Dobrushin curve for the dephasing channel $\Phi_{\pi/4}$ with rotated principal axes. The algorithm required an average of $538$ elements (worst case: $730$ elements) in the partition $\cP$ to reach the desired precision.][The Dobrushin curve for the dephasing channel\\ $\Phi_{\pi/4}$ with rotated principal axes. The algorithm required an average of $598$ elements (worst case: $730$ elements) in the partition $\cP$ to reach the desired precision. \label{fig:dobrushinrotated2}]
{\includegraphics[width=0.5\linewidth]{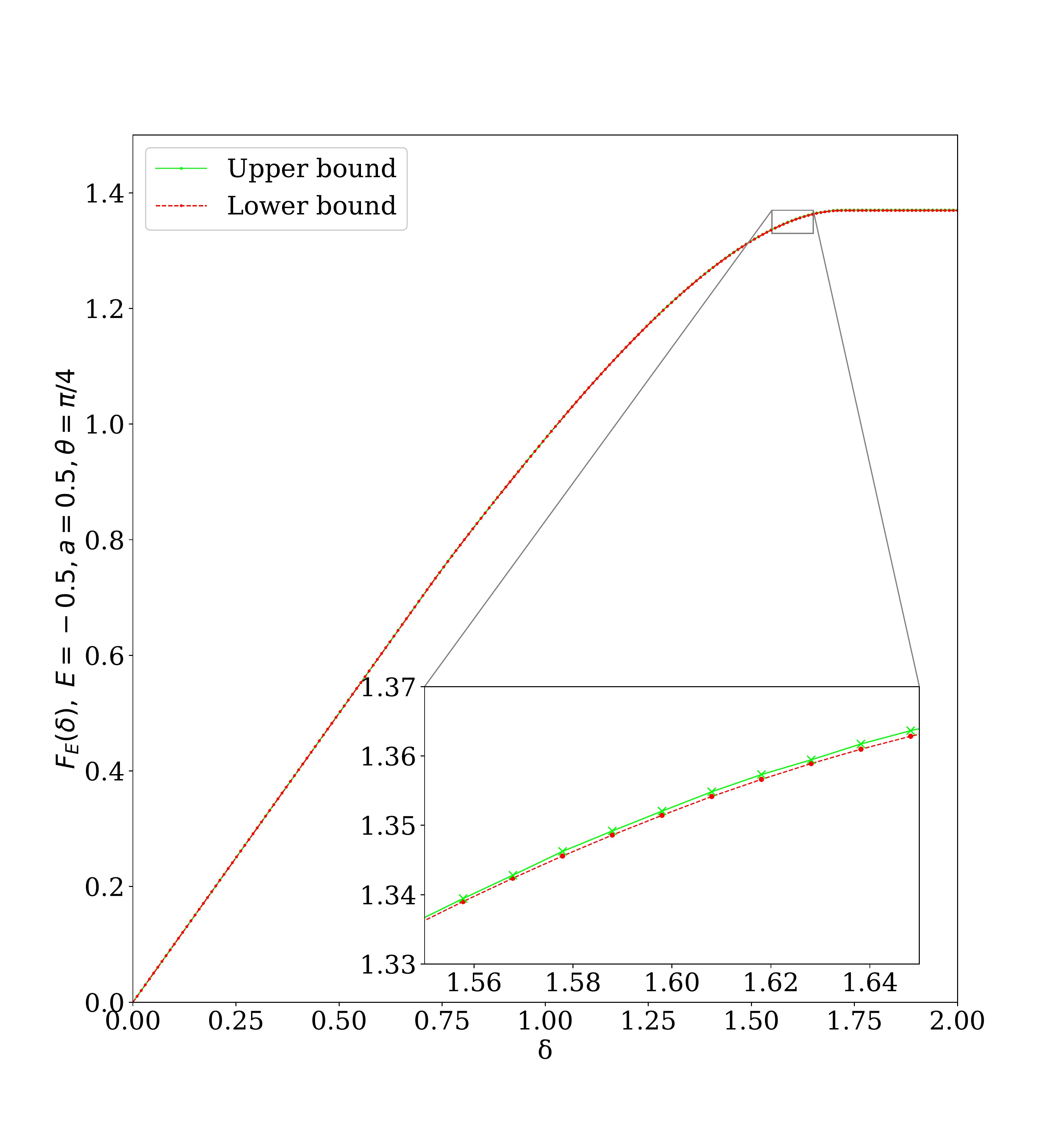}}
\caption{The Dobrushin curves of the  dephasing channels $\Phi_{\pi/9}$ and $\Phi_{\pi/4}$
with rotated principal axes for two different values of the rotation angle $\theta$, for $a = \frac{1}{2}$ and $E = -0.5$. For $200$ values of $\delta \in [0,2]$, the algorithm was run with a desired precision of $\epsilon = 10^{-3}$.\label{fig:dobrushindephasingrotated}}
\end{figure}

Finally, for  $\delta>2\sqrt{1-|E|^2}$ the  two inputs
\begin{align}
\vec{r}_1&=(2\sqrt{1-|E|^2},0,-1+|E|)\ ,\\
\vec{r}_2&=(-2\sqrt{1-|E|^2},0,-1+|E|)
\end{align}
are optimal and lead to $\|\vec{r}'_1-\vec{r}'_2\|_1=a\cdot\|\vec{r}_1-\vec{r}_2\|_1$. 
\end{proof}

\subsubsection{Dobrushin-curves of generic qubit channels\label{sec:genericqubitdobrushin}}
In Fig.~\ref{fig:dobrushindephasingrotated}, we consider 
more general channels which no do not obey the symmetry condition~\eqref{eq:invariancephi}.

Consider a dephasing channel whose principal axes are not the $\{\sigma_1, \sigma_2, \sigma_3\}$-axes. To achieve this, we rotate the Kraus operators of the channel around the $\sigma_1$-axis by an angle $\theta_1\in[0,2\pi]$. This means that we conjugate by the unitary $e^{\frac{i}{2}\theta_1 \sigma_1}$, obtaining the channel
\begin{equation}
  \Phi_{\theta_1}(\rho) = e^{-\frac{i}{2}\theta_1\sigma_1} \Phi\left(e^{\frac{i}{2}\theta_1\sigma_1} \rho e^{-\frac{i}{2}\theta_1\sigma_1}\right) e^{\frac{i}{2}\theta_1\sigma_1} \ ,
\end{equation}
where $\Phi$ is the dephasing channel (see Eq.~\eqref{eq:dephasingdef}).
The Dobrushin curve of such a channel for a fixed $\theta_1$ can be calculated by our algorithm and results in the curve depicted in Fig.~\ref{fig:dobrushindephasingrotated}.

\subsubsection*{Program code}
Python program code for the algorithm constructed here is available together with the TeX-source code on the ArXiv. 

\subsubsection*{Acknowledgments}
SH and RK acknowledge support by  DFG project no.~K05430/1-1. 
RK is supported by the Technische Universit\"{a}t M\"{u}nchen - Institute of Advanced Study, funded by the German Excellence Initiative and the European Union Seventh Framework Programme under grant agreement no.~291763.
MT acknowledges an Australian Research Council Discovery Early Career Researcher Award, project no.~DE160100821. 
The authors thank Volkher Scholz and David Reeb for discussions on Dobrushin curves.


\newpage
\appendix

\section{Pseudocode for jointly constrained semidefinite bilinear programs}
\label{app:pseudo}

In this appendix we give pseudocode for the full algorithm.

  \begin{figure}[h!]
  \begin{algorithmic}[1]
  \Statex
  \Function{ComputeVectorRep}{$Q,X,Y,\{\eta_j\}_{j=1}^{p^2}, \{\xi_j\}_{j=1}^{q^2}$}\\ \hrulefill\\
  \textbf{Input:} Operators $Q \in \selfadjointops(\mathbb{C}^p \otimes \mathbb{C}^q), X \in \selfadjointops(\mathbb{C}^p), Y \in \selfadjointops(\mathbb{C}^q)$\\
  orthonormal bases $\{\eta_j\}_{j=1}^{p^2}$
  of $\selfadjointops(\mathbb{C}^p)$ and
  $\{\xi_k\}_{k=1}^{q^2}$ of   $\selfadjointops(\mathbb{C}^q)$\\
  \textbf{Output:} Vectors $(x,y)\in\mathbb{R}^{p^2}\times \mathbb{R}^{q^2}$ such that $(x,y)=\Gamma(X,Y)$\\ \hrulefill \\
   \State{Find $U_{j,k} = \tr((\eta_j \otimes \xi_k) Q)$ \textbf{ for } $j = 1, \dots, p^2, k = 1, \dots, q^2$}
    \State{Find singular value decomposition $U = S \Delta T$}
            \For{ each $j=1,\ldots,p^2$}
          \State{Set $x_j=\sum_{k=1}^{p^2} S_{k,j}\tr(X \eta_k)$}
                   \EndFor
           \For{ each $k=1,\ldots,q^2$}
          \State{Set $y_k=\sum_{\ell=1}^{q^2} T_{\ell,k}\tr(Y \xi_\ell)$}
           \EndFor
     \State{\Return $(x,y)$}
   \EndFunction
 \end{algorithmic}

  \begin{algorithmic}[1]
  \Statex
  \Function{ComputeOperator}{$Q,x,y,\{\eta_j\}_{j=1}^{p^2}, \{\xi_j\}_{j=1}^{q^2}$}\\ \hrulefill\\
  \textbf{Input:} Operators $Q \in \selfadjointops(\mathbb{C}^p \otimes \mathbb{C}^q), (x,y)\in\mathbb{R}^{p^2}\times\mathbb{R}^{q^2}$\\
  orthonormal bases $\{\eta_j\}_{j=1}^{p^2}$
  of $\selfadjointops(\mathbb{C}^p)$ and
  $\{\xi_k\}_{k=1}^{q^2}$ of   $\selfadjointops(\mathbb{C}^q)$\\
  \textbf{Output:} Operators $(X,Y)\in\selfadjointops(\mathbb{C}^p)\times\selfadjointops(\mathbb{C}^q)$ such that $(x,y)=\Gamma(X,Y)$\\ \hrulefill \\
   \State{Find $U_{j,k} = \tr((\eta_j \otimes \xi_k) Q)$ \textbf{ for } $j = 1, \dots, p^2, k = 1, \dots, q^2$}
    \State{Find singular value decomposition $U = S \Delta T$}
                     \State{Set $X=\sum_{j=1}^{p^2}\sum_{k=1}^{p^2} S_{k,j}x_k \eta_j$}
                           \State{Set $Y=\sum_{k=1}^{q^2}\sum_{\ell=1}^{q^2} T_{\ell,k}y_\ell \xi_k$}
     \State{\Return $(X,Y)$}
   \EndFunction
 \end{algorithmic}

 \caption{The algorithm \textsc{ComputeVectorRep}
computes the image $\Gamma(X,Y)$
of two operators $(X,Y)$, cf. Lemma~\ref{lem:vectorproblem}.
Conversely, the procedure \textsc{ComputeOperator} computes $\Gamma^{-1}(x,y)$, i.e, converts a pair of vectors into a pair of operators.
We note that the required singular value decomposition can be computed once and stored. It can then be reused whenever the function is invoked. \label{alg:computevectorrep}}
\end{figure}
  \begin{figure}[h]
    \begin{algorithmic}[1]
      \Function{BranchHyperrectangle}{$\Omega, (v,w)$}\\ \hrulefill\\
      \textbf{Input:} Hyperrectangle $\Omega = \Omega(\ell,L,m,M) \subset\mathbb{R}^{p^2} \times \mathbb{R}^{q^2}$, branching point $(v,w) \in \mathbb{R}^{p^2} \times \mathbb{R}^{q^2}$ \\
      \textbf{Output:} Hyperrectangles $\Omega^{(1)}, \Omega^{(2)}, \Omega^{(3)}, \Omega^{(4)} \subset \mathbb{R}^{p^2} \times \mathbb{R}^{q^2}$ such that $\cup_{i=1}^4 \Omega^{i} = \Omega$ \\ \hrulefill \\
      \State{Pick index $I \in \{1, \dots, K\}$ which produces the largest difference between the two sides in the inequality}
      \State{$\max \{m_i v_i + \ell_i w_i - \ell_i m_i, M_i v_i + L_i w_i - L_i M_i\} < v_i w_i$}
      \For{$i \gets 1, \dots, K$}
      \If{$i \neq I$}
      \For{$j \gets 1, 2, 3, 4$}
      \State{$(\ell_i^{(j)}, L_i^{(j)}, m_i^{(j)}, M_i^{(j)}) \gets (\ell_i, L_i, m_i, M_i)$}\EndFor \EndIf
      \State{$(\ell_I^{(1)}, L_I^{(1)}, m_I^{(1)}, M_I^{(1)}) \gets (\ell_I, v_I, m_I, w_I)$} 
      \State{$(\ell_I^{(2)}, L_I^{(2)}, m_I^{(2)}, M_I^{(2)}) \gets (v_I, L_I, m_I, w_I)$}
      \State{$(\ell_I^{(3)}, L_I^{(3)}, m_I^{(3)}, M_I^{(3)}) \gets (v_I, L_I, w_I, M_I)$}
      \State{$(\ell_I^{(4)}, L_I^{(4)}, m_I^{(4)}, M_I^{(4)}) \gets (\ell_I, v_I, w_I, M_I)$}
      \EndFor
      \For{$i \gets K+1, \dots, p^2$} \For{$j \gets 1,2,3,4$}
      \State{$(\ell_i^{(j)}, L_i^{(j)}) \gets (\ell_i, L_i)$} \EndFor\EndFor
      \For{$i \gets K+1, \dots, q^2$} \For{$j \gets 1,2,3,4$}
      \State{$(m_i^{(j)}, M_i^{(j)}) \gets (m_i, M_i)$} \EndFor\EndFor
      \State{\Return $\Omega(\ell^{(1)}, L^{(1)}, m^{(1)}, M^{(1)}),  \Omega(\ell^{(2)}, L^{(2)}, m^{(2)}, M^{(2)}),  \Omega(\ell^{(3)}, L^{(3)}, m^{(3)}, M^{(3)}),  \Omega(\ell^{(4)}, L^{(4)}, m^{(4)}, M^{(4)})$}
      \EndFunction
    \end{algorithmic}
  \caption{This algorithm splits a hyperrectangle~$\Omega$ into four subrectangles, cf.~Fig.~\ref{fig:hyperrectanglesplit}.
  The choice of coordinates $(x_I,y_I)$ in the first step is restricted to
  $I\leq K$ (where $K$ is the number of non-zero singular values of $U$ as in Eq.~\eqref{eq:svd}). 
  The hyperrectangles do not need to be refined along the remaining coordinates because the objective function is linear in these. In particular, the dependence of the convex envelope~$\Vex_\Omega F$ on the parameters $\{v_j\}_{j=K+1}^{p^2}$ and $\{w_j\}_{j=K+1}^{q^2}$ coincides with that of the function~$F$ irrespective of which hyperrectangle~$\Omega$ is considered.  \label{alg:branchhyperrectangle}}
  \end{figure}

\begin{figure}
  \begin{algorithmic}[1]
  \Statex
  \Function{Boundingrectangle}{$Q,\cS,\{\eta_j\}_{j=1}^{p^2}, \{\xi_k\}_{k=1}^{q^2}$}\\ \hrulefill\\
  \textbf{Input:} Operators $Q \in \selfadjointops(\mathbb{C}^p \otimes \mathbb{C}^q)$, set $\cS \subset \selfadjointops(\mathbb{C}^p) \times \selfadjointops(\mathbb{C}^q)$ defined by SDP constraints\\
  orthonormal bases $\{\eta_j\}_{j=1}^{p^2}$
  of $\selfadjointops(\mathbb{C}^p)$ and $\{\xi_k\}_{k=1}^{q^2}$
  of $\selfadjointops(\mathbb{C}^q)$ \\
  \textbf{Output:} $\ell^*,L^*\in\mathbb{R}^{p^2}$ and $m^*,M^*\in\mathbb{R}^{q^2}$  such that $\Gamma(\cS)\subset \Omega(\ell^*,L^*,m^*,M^*)$ and $\Omega$ is minimal as discussed in Lemma~\ref{lem:boundinghyperrectangle}\\ \hrulefill \\

   \State{Find $U_{j,k} = \tr((\eta_j \otimes \xi_k) Q)$ \textbf{ for } $j = 1, \dots, p^2, k = 1, \dots, q^2$}
    \State{Find singular value decomposition $U = S \Delta T$}

      \For{ each $j=1,\ldots,p^2$}
          \State{Use the SDP solver to compute $\ell_j^*=\inf_{(X,Y)\in\cS} \sum_{k=1}^{p^2} S_{k,j}\tr(X \eta_k)$}
                    \State{Use the SDP solver to compute $L_j^*=\sup_{(X,Y)\in\cS} \sum_{k=1}^{p^2} S_{k,j}\tr(X \eta_k)$}
      \EndFor
           \For{ each $k=1,\ldots,q^2$}
          \State{Use the SDP solver to compute $m_k^*=\inf_{(X,Y)\in\cS} \sum_{\ell=1}^{q^2} T_{\ell,k}\tr(Y \xi_\ell)$}
                    \State{Use the SDP solver to compute $M_k^*=\sup_{(X,Y)\in\cS} \sum_{\ell=1}^{q^2} T_{\ell,k}\tr(Y \xi_\ell)$}
      \EndFor
\State{\Return $(\ell^*,L^*,m^*,M^*)$}
   \EndFunction
 \end{algorithmic}
 \caption{The procedure \textsc{BoundingRectangle} which
finds a minimal hyperrectangle~$\Omega$ containing the vectors~$\Gamma(\cS)$. It invokes an SDP solver. Again, the singular value decomposition of $Q$ can be precomputed. \label{alg:boundingrec}}
\end{figure}

\begin{figure}
    \begin{algorithmic}[1]
      \Statex
    \Function{ComputeBoundsSDP}{$f,\Omega,\cS,\{\eta_j\}_{j=1}^{p^2},\{\xi_k\}_{k=1}^{q^2}$}\\ \hrulefill\\
    \textbf{Input:}   a function $f: \mathbb{R}^{p^2}\times\mathbb{R}^{q^2} \rightarrow \mathbb{R}$ of the form~\eqref{eq:fvectorformexplicit}\\
     hyperrectangle $\Omega \subset\mathbb{R}^{p^2} \times \mathbb{R}^{q^2}$ and\\
     set $\cS \subset \selfadjointops(\mathbb{C}^p) \times \selfadjointops(\mathbb{C}^q)$ defined by SDP constraints\\
  orthonormal bases $\{\eta_j\}_{j=1}^{p^2}$
  of $\selfadjointops(\mathbb{C}^p)$ and $\{\xi_k\}_{k=1}^{q^2}$
  of $\selfadjointops(\mathbb{C}^q)$ \\
    \textbf{Output:} $(\underline{\alpha}(\Omega),\overline{\alpha}(\Omega))\in\mathbb{R}^2$ and $z(\Omega)=(x^*, y^*) \in \Gamma(\cS)\cap \Omega$
    such that\\ $\underline{\alpha}(\Omega)\leq \min_{(x,y)\in\Gamma(\cS)\cap \Omega}f(x,y)\leq \overline{\alpha}(\Omega)=f(x^*,y^*)$ \\ \hrulefill

    \State{Define $\hat{\ell}_j^0=\ell_j$ and $\hat{\ell}_j^1=L_j$ for $j=1,\ldots,p^2$}
     \State{Define $\hat{m}_k^0=m_k$ and $\hat{m}_k^1=M_k$ for $k=1,\ldots,q^2$}
     \State{Define the operators and scalars
     \begin{align*}
     G^b_j=\sigma_j \hat{m}_j^{b}\sum_{k=1}^{p^2} S_{k,j} \eta_k\qquad H^b_j=\sigma_j \hat{\ell}_j^{b}\sum_{\ell=1}^{q^2} T_{j,\ell} \xi_\ell\qquad s^b_j=\sigma_j \hat{\ell}_j^b \hat{m}_j^b\end{align*}
      for all $b\in \{0,1\}$ and $j=1,\ldots,K$.}
     \State{Define the function $G:\cB(\mathbb{C}^p)\times \cB(\mathbb{C}^q)\times \mathbb{R}^K\rightarrow\mathbb{R}$ by
     \begin{align*}
     G(X,Y,r)=\sum_{j=1}^K r_j + \sum_{j=1}^{p^2} a_j
\sum_{k=1}^{p^2} S_{k,j}\tr(X\eta_k)
+ \sum_{\ell=1}^{q^2} b_j T_{\ell,j}\tr(Y\xi_\ell)
\end{align*}}
        \State{Invoke the SDP solver to compute
\begin{align*}
(X^*,Y^*,r^*)=\argmin_{(X,Y,r)}
G(X,Y,r)
\end{align*}\\
subject to the constraints
\begin{align*}
(X,Y)&\in\cS\\
   \ell_j&\leq  \tr(X \sum_{k=1}^{p^2} S_{k,j}\eta_k)\leq L_j\qquad\textrm{ for }\qquad j=1,\ldots,p^2\qquad\textrm{and }\\
  m_k&\leq \tr(Y\sum_{\ell=1}^{q^2} T_{k,\ell}\xi_\ell)\leq M_k\qquad\textrm{ for }\qquad k=1,\ldots,q^2\\
    \tr(XG^b_j)+\tr(YH^b_j)-r_j&\leq s^b_j\qquad\textrm{ for all }b\in \{0,1\}\qquad\textrm{ and }\qquad j=1,\ldots,K\ .
\end{align*}
}
      \State{Set $\underline{\alpha}(\Omega)=G(X^*,Y^*,r^*)$}
        \State{Set $\overline{\alpha}(\Omega)=f(\Gamma(X^*,Y^*))$}
    \State{\Return{$\underline{\alpha}(\Omega)$, $\overline{\alpha}(\Omega)$ and $(x(X^*),y(Y^*))$ }}
      \EndFunction
    \end{algorithmic}
    \caption{The subroutine \textsc{ComputeBoundsSDP} takes as input a set $\Omega\subset\mathbb{R}^{p^2}\times\mathbb{R}^{q^2}$ and a  function~$f:\Omega\rightarrow\mathbb{R}$ as in Lemma~\ref{lem:vectorproblem}. It further accepts a set $\cS\subset\cB(\mathbb{C}^p)\times\cB(\mathbb{C}^q)$ defined by SDP constraints. It computes a pair $(\underline{\alpha}(\Omega),\overline{\alpha}(\Omega))$ of lower and upper bounds on $\inf_{(x,y)\in \Gamma(\cS)\cap \Omega}f(x,y)$, where $\Gamma$ is defined as in Lemma~\ref{lem:vectorproblem}.
   \label{alg:computeboundssdp}
  }
  \end{figure}

 \begin{figure}
  \begin{algorithm}[H]
    \caption{Jointly constrained semidefinite bilinear programming}
    \begin{algorithmic}[1]
      \Statex
      \textbf{Input:}
      $\cS \subset \selfadjointops(\mathbb{C}^p) \times \selfadjointops(\mathbb{C}^q)$ specified by SDP constraints.
      $Q \in \selfadjointops(\mathbb{C}^p \otimes \mathbb{C}^q), A \in \selfadjointops(\mathbb{C}^p), B \in \selfadjointops(\mathbb{C}^q)$
      determining  a function $F:\selfadjointops(\mathbb{C}^p)\times\selfadjointops(\mathbb{C}^q)\rightarrow\mathbb{R}$ as in Eq.~\eqref{eq:sdpproblem}\\
        desired precision $\epsilon > 0$.
      \\
      \textbf{Output:} $(X^*,Y^*) \in \cS$ and $\overline{\alpha}$ such that Eq.~\eqref{eq:outputalgnew} holds \\ \hrulefill
      \State{Fix bases $\{\eta_j\}_{j=1}^{p^2}$
  of $\selfadjointops(\mathbb{C}^p)$ and $\{\xi_k\}_{k=1}^{q^2}$
  of $\selfadjointops(\mathbb{C}^q)$}
      \State{Set $(\ell^*,L^*,m^*,M^*)=\Call{BoundingRectangle}{Q,\cS,\{\eta_j\}_{j=1}^{p^2}, \{\xi_k\}_{k=1}^{q^2}}$}
      \State{Define $D$ as the hyperrectangle $D=\Omega(\ell^*,L^*,m^*,M^*)$}
         \State{Set $(\underline{\alpha}(D),\overline{\alpha}(D),z(D))=
         \Call{ComputeBoundsSDP}{f,D,\cS,\{\eta_j\}_{j=1}^{p^2},\{\xi_k\}_{k=1}^{q^2}}$}
       \State{Set $\cP=\{D\}$, $\underline{\alpha}(\cP)=\underline{\alpha}(D)$ and $\overline{\alpha}(\cP)=\overline{\alpha}(D)$}
      \While{$\overline{\alpha}(\cP)-\overline{\alpha}(\cP)>\epsilon$}
        \State{Find (any) hyperrectangle~$\Omega\in\cP$ such that $\underline{\alpha}(\cP)=\underline{\alpha}(\Omega)$}
        \State{Set $(x,y)=z(\Omega)$}
        \State{$(\Omega^{(1)}, \Omega^{(2)}, \Omega^{(3)}, \Omega^{(4)}) =$ \Call{BranchHyperrectangle}{$\Omega,(x,y)$}}
         \For{$j \gets 1,2,3,4$}
             \State{Set $(\underline{\alpha}(\Omega^{(j)}),\overline{\alpha}(\Omega^{(j)}),z(\Omega^{(j)}))=
         \Call{ComputeBoundsSDP}{f,\cS,(\Omega^{(j)})}$}
         \EndFor
         \State{Update $\cP\gets (\cP\backslash \{\Omega\})\cup \{ \Omega^{(1)}, \Omega^{(2)}, \Omega^{(3)}, \Omega^{(4)}\}$}
         \State{Compute $\underline{\alpha}(P)=\min_{\Omega\in\cP}\underline{\alpha}(\Omega)$ and  $\overline{\alpha}(P)=\min_{\Omega\in\cP}\overline{\alpha}(\Omega)$}
      \EndWhile
      \State{Find (any) hyperrectangle~$\Omega\in\cP$ such that
      $\overline{\alpha}(\Omega)=\overline{\alpha}(\cP)$}
                \State{Terminate and \Return $\Call{ComputeOperator}{z(\Omega)}$, $\overline{\alpha}(\cP)$}
      \end{algorithmic}
  \end{algorithm}
  \caption{The algorithm for jointly constrained biconvex programming.
    The key modification compared to the biconvex programming algorithm from Section~\ref{sec:biconvex} is the use of the  subroutine \textsc{ComputeBoundsSDP}, which   uses an SDP solver to establish bounds on the objective function.
   \label{alg:branch-and-boundsdp}}
  \end{figure}

\end{document}